\documentclass[runningheads]{llncs}

\usepackage{ifthen}
\newif\ifextendedversion

\extendedversiontrue    

\usepackage{graphicx}
\usepackage{hyperref}
\usepackage{xcolor}

\usepackage{xspace}
\usepackage{booktabs}
\usepackage{amssymb}
\usepackage{savesym}
\usepackage{stmaryrd}
\usepackage{mathtools}
\usepackage{enumerate}
\usepackage{marvosym}

\usepackage{dutchcal} 

\def\domi#1{{\tt{doi:}}\href{http://dx.doi.org/#1}{\nolinkurl{#1}}}
\def\murl#1{{\tt{url:}}\href{#1}{\nolinkurl{#1}}}
\newcounter{Romancnt}

\usepackage{rotating}

\usepackage{tikz}
\usetikzlibrary{arrows,arrows.meta}
\usetikzlibrary{backgrounds,positioning}
\usetikzlibrary{decorations}
\usetikzlibrary{decorations.markings}
\usetikzlibrary{decorations.pathmorphing}

\tikzset{pstep/.style={decoration={markings, mark=at position .5 with with%
{\node [transform shape] {$\prsymbol$};}},postaction={decorate}}}
\tikzset{mstep/.style={decoration={markings, mark=at position .5 with with%
{\node [midway,solid,thin,circle,inner sep=1.5,draw] {};}},postaction={decorate}}}
\tikzset{>={To[length=.9mm,angle'=60]}}

\definecolor{darkblue}{rgb}{0,0.00,0.66}

\makeatletter
\def\lubd#1#2#3{%
 \def\next##1##2{%
  \setbox0=\hbox{$##1\vphantom{#3}\@ifempty{#1}{}{_{\vphantom{#1}}}%
   \@ifempty{#2}{}{^{#2}}$}%
  \setbox1=\hbox{$##1\vphantom{#3}\@ifempty{#1}{}{_{#1}}%
   \@ifempty{#2}{}{^{\vphantom{#2}}}$}%
  \setbox2=\vbox{\hbox to\wd0{}\hbox to\wd1{}}%
  {\hskip\wd2\hskip-\wd0\box0\hskip-\wd1\box1{#3}}%
 }%
 \mathpalette\next{}%
}
\def\lud#1#2#3{{\lubd{#1}{#2}{#3}}}
\def\rubd#1#2#3{#3\@ifempty{#1}{}{_{#1}}\@ifempty{#2}{}{^{#2}}}
\def\rubd#1#2#3{{\rubd{#1}{#2}{#3}}}
\makeatother
\def\op#1#2{#1^{#2}}
\def\ip#1#2{#1_{#2}}
\def\zp#1#2{\lud{#2}{}{#1}}
\def\ep#1#2{\lud{}{#2}{#1}}
\def\ezp#1#2#3{\lud{#3}{#2}{#1}}

\catcode`\@=11
\newcommand{\wash}{\mathpalette\mywash}
\newcommand{\mywash}[2]{\setbox0=\hbox{$\m@th#1{#2}$}\wd0=0pt\box0}
\catcode`\@=12
\newcommand{\arsdev}{\mathrel{\wash{\,\,\,\circ}\mathord{\longrightarrow}}}
\newcommand{\arsdevinv}{\mathrel{\wash{\,\,\,\,\circ}\mathord{\longleftarrow}}}

\def\isdefd{=} 

\def\hastype{\mathbin{:}} 

\def\relap#1#2#3{#2\mathrel{#1}#3} 
\def\binap#1#2#3{#2\mathbin{#1}#3} 
\def\funap#1#2{#1(#2)} 

\def\pairsep{{,}}
\def\pair#1#2{(#1\pairsep#2)} 
\def\triple#1#2#3{\pair{#1}{#2\pairsep#3}} 

\def\ssetin{{\in}}  
\def\setin{\binap{\ssetin}}

\def\setstr#1{\{#1\}} 
\def\setabs#1#2{\setstr{#1\mathrel{|}#2}} 
\def\ssetle{{\subseteq}} 
\def\setle{\relap{\ssetle}}

\def\sseteq{{=}} 
\def\seteq{\relap{\sseteq}}
\def\ssetge{{\supseteq}} 
\def\setge{\relap{\ssetge}}
\def\ssetneq{{\neq}} 
\def\setneq{\relap{\ssetneq}}
\def\setemp{{\emptyset}} 
\def\ssetlub{{\cup}} 
\def\setlub{\binap{\ssetlub}}
\def\ssetLub{\bigcup} 
\def\setLub#1#2{\ssetLub_{#1}{#2}}
\def\ssetglb{{\cap}} 
\def\setglb{\binap{\ssetglb}}

\def\ssetdif{{-}} 
\def\setdif{\binap{\ssetdif}}

\def\srelref{{=}} 
\def\relref#1{\op{#1}{\srelref}}
\def\sreltra{{+}} 
\def\reltra#1{\op{#1}{\sreltra}}
\def\sreltre{{\ast}} 
\def\reltre#1{\op{#1}{\sreltre}}

\def\srelsco{{\cdot}} 
\def\relsco{\binap{\srelsco}}

\def\srelle{{\ssetle}} 
\def\relle{\relap{\srelle}}
\def\sreleq{\sseteq} 
\def\releq{\relap{\sreleq}}

\def\relnco#1#2{\op{#2}{#1}}

\def\aord{{\succ}}
\def\orda{\relap{\aord}}
\def\aordref{{\succeq}}

\def\aordmul{\ip{\aord}{\mul}}

\def\aordmulref{\ip{\aordref}{\mul}}
\def\ordmulrefa{\relap{\aordmulref}}
\def\aordhot{\hot{\aord}}
\def\ordhota{\relap{\aordhot}}
\def\aordrefhot{\aordhotref}
\def\ordhotrefa{\ordrefhota}
\def\aordrefhot{\hot{\aordref}}
\def\ordrefhota{\relap{\aordrefhot}}
\def\aorddwn{{\curlyvee}}
\def\orddwna#1{{\aorddwn#1}}

\def\ordhotrefdwna#1{{\aordhotrefdwn#1}}
\def\aordhotrefdwn{{{\shortmid}\!\aordhotdwn}}
\def\aordhotdwn{\hot{\aorddwn}}
\def\ordhotdwna#1{{\aordhotdwn#1}}

\def\natzer{{0}}
\def\natone{{1}}
\def\nattwo{{2}}
\def\natthr{{3}}
\def\natfou{{4}}
\def\anat{{n}}
\def\bnat{{m}}

\def\snatgt{{>}}
\def\natgt{\relap{\snatgt}}

\def\snateq{{=}}
\def\nateq{\relap{\snateq}}
\def\snatneq{{\neq}}
\def\natneq{\relap{\snatneq}}

\def\snatge{{\geqslant}}
\def\natge{\relap{\snatge}}

\def\snatgge{\raisebox{.06em}{$\lefteqn{\kern.06em\underline{\phantom{\kern.1em>}}}{\gg}$}}

\def\natggt{\relap{\natggt}}

\def\snatdif{{-}} 
\def\natdif{\binap{\snatdif}}

\def\snatSum{{\sum}}
\def\natSum#1#2{\ip{\snatSum}{#1}#2}

\def\aidx{{i}}
\def\bidx{{j}}
\def\cidx{{k}}
\def\idxin{1}
\def\idxout{0}

\def\aars{{\to}}
\def\arsa{\relap{\aars}}
\def\aiars{\ip{\aars}}
\def\iarsa#1{\relap{\aiars{#1}}}
\def\aarsref{\relref{\aars}}

\def\aarstra{\reltra{\aars}}

\def\aiarstra{\ip{\aarstra}}

\def\aarsinv{{\leftarrow}}
\def\arsinva{\relap{\aarsinv}}
\def\aiarsinv{\zp{\aarsinv}}

\def\aiarsinvref#1{\lub{#1}{\srelref}{\aarsinv}}

\def\iarsinva#1{\relap{\aiarsinv{#1}}}
\def\aarsnco#1{\relnco{#1}{\aars}}

\def\aiarsnco#1{\ip{\aarsnco{#1}}}

\def\aiarsinvnce{\ezp{\aarsinv}}

\def\aarsdev{{\arsdev}}
\def\arsdeva{\relap{\aarsdev}}
\def\aiarsdev{\ip{\aarsdev}}
\def\iarsdeva#1{\relap{\aiarsdev{#1}}}
\def\aarsdevinv{{\arsdevinv}}

\def\arsdevinva{\relap{\aarsdevinv}}
\def\aiarsdevinv{\zp{\aarsdevinv}}
\def\iarsdevinva#1{\relap{\aiarsdevinv{#1}}}

\def\bars{%
  \mathchoice
  {\tikz[baseline] \draw[line width=0.3pt,-{Latex[open]},yshift=0.60ex] (0,0) -- (1.1em,0);}
  {\tikz[baseline] \draw[line width=0.3pt,-{Latex[open]},yshift=0.60ex] (0,0) -- (1.1em,0);}
  {\tikz[baseline] \draw[line width=0.3pt,-{Latex[open]},yshift=0.43ex] (0,0) -- (1.1em,0);}
  {\tikz[baseline] \draw[line width=0.3pt,-{Latex[open]},yshift=0.38ex] (0,0) -- (1.1em,0);}
}
\def\arsb{\relap{\bars}}

\def\barsinv{%
  \mathchoice
  {\tikz[baseline] \draw[line width=0.3pt,{Latex[open]}-,yshift=0.60ex] (0,0) -- (1.1em,0);}
  {\tikz[baseline] \draw[line width=0.3pt,{Latex[open]}-,yshift=0.60ex] (0,0) -- (1.1em,0);}
  {\tikz[baseline] \draw[line width=0.3pt,{Latex[open]}-,yshift=0.43ex] (0,0) -- (1.1em,0);}
  {\tikz[baseline] \draw[line width=0.3pt,{Latex[open]}-,yshift=0.38ex] (0,0) -- (1.1em,0);}
}
\def\arsinvb{\relap{\barsinv}}

\def\aarstre{{\twoheadrightarrow}}
\def\arstrea{\relap{\aarstre}}

\def\aiarstre{\ip{\aarstre}}
\def\iarstrea#1{\relap{\aiarstre{#1}}}

\def\aarssym{{\leftrightarrow}}

\def\aarsequ{\reltre{\aarssym}}
\def\arsequa{\relap{\aarsequ}}
\def\aiarsequ{\ip{\aarsequ}}
\def\iarsequa#1{\relap{\aiarsequ{#1}}}

\def\barsequ{\reltre{\barssym}}

\def\barssym{{\barsinv\kern-.485em\bars}}

\def\aobj{{a}} 
\def\bobj{{b}}
\def\cobj{{c}}

\def\aiobj{\ip{\aobj}} 

\def\sobjeq{{=}} 
\def\objeq{\relap{\sobjeq}}

\def\astp{{\phi}} 
\def\bstp{{\psi}}

\def\cstp{{\omega}}
\def\aistp{\ip{\astp}}
\def\bistp{\ip{\bstp}}

\def\aStp{{\Phi}} 
\def\bStp{{\Psi}}
\def\cStp{{\Omega}}

\def\aiStp{\ip{\aStp}}
\def\biStp{\ip{\bStp}}

\def\sstpsrc{\mathsf{src}}

\def\sstptgt{\mathsf{tgt}}

\def\sstpin{{\sclsin}} 
\def\stpin{\binap{\sstpin}}

\def\sstple{\ssetle} 
\def\stple{\relap{\sstple}}
\def\sstpres{{/}} 
\def\stpres#1#2{#1\sstpres#2}

\def\asym{{f}} 
\def\bsym{{g}}
\def\csym{{h}}

\def\aRul{{R}}
\def\bRul{{S}}

\def\dRul{{C}}
\def\arul{{\varrho}}
\def\brul{{\theta}}
\def\crul{{\eta}}
\def\airul{\ip{\arul}}
\def\birul{\ip{\brul}}
\def\cirul{\ip{\crul}}
\def\rula{\funap{\arul}}
\def\rulb{\funap{\brul}}
\def\rulc{\funap{\crul}}
\def\irula#1{\funap{\airul{#1}}}
\def\irulb#1{\funap{\birul{#1}}}
\def\irulc#1{\funap{\cirul{#1}}}
\def\srulstr{{\to}}
\def\rulstr{\binap{\srulstr}}

\def\atrm{{t}}
\def\btrm{{s}}
\def\ctrm{{u}}
\def\dtrm{{w}}
\def\etrm{{p}}

\def\atrmp{\hat{\atrm}}
\def\amtrm{{M}}
\def\bmtrm{{N}}
\def\cmtrm{{L}}

\def\alhs{{\ell}}
\def\blhs{{g}}

\def\arhs{{r}}
\def\brhs{{d}}

\def\ailhs{\ip{\alhs}}
\def\airhs{\ip{\arhs}}
\def\bilhs{\ip{\blhs}}
\def\birhs{\ip{\brhs}}
\def\aitrm{\ip{\atrm}}
\def\bitrm{\ip{\btrm}}
\def\citrm{\ip{\ctrm}}

\def\aimtrm{\ip{\amtrm}}

\def\strmeq{{=}} 
\def\trmeq{\relap{\strmeq}}
\def\strmneq{{\neq}} 
\def\trmneq{\relap{\strmneq}}

\def\atrmvar{{x}}
\def\btrmvar{{y}}
\def\ctrmvar{{z}}
\def\dtrmvar{{w}}

\def\aitrmvar{\ip{\atrmvar}}
\def\bitrmvar{\ip{\btrmvar}}
\def\citrmvar{\ip{\ctrmvar}}
\def\ditrmvar{\ip{\dtrmvar}}

\def\varmvar{\uppercase}
\def\amtrmvar{\varmvar\atrmvar} 
\def\bmtrmvar{\varmvar\btrmvar}
\def\cmtrmvar{\varmvar\ctrmvar}

\def\aimtrmvar{\ip{\amtrmvar}} 
\def\bimtrmvar{\ip{\bmtrmvar}}

\def\mtrmvara{\funap{\amtrmvar}} 
\def\mtrmvarb{\funap{\bmtrmvar}}
\def\mtrmvarc{\funap{\cmtrmvar}}

\def\imtrmvarb#1{\funap{\bimtrmvar{#1}}}

\def\trmencref{\relap{\makebox[0pt]{\makebox[7.5pt][r]{\raise 0.6pt \hbox{$\cdot$}}}{\trianglerighteq}}}
\def\strmencref{{\trmencref{}{}}}

\def\strmenc{{\lefteqn{\cdot}\triangleright}} 
\def\trmenc{\relap{\strmenc}}

\def\strmsbmref{{\wash{\leq}{\raise .625pt\hbox{$\lessdot$}}}} 

\def\strmsbmlub{{\wash{\vee}\kern.185em\raisebox{.4ex}{$\cdot$}}} 

\def\strmsbmglb{{\wash{\wedge}\kern.185em\raisebox{-.4ex}{$\cdot$}}} 

\def\trmsiz#1{{\parallel}#1{\parallel}}

\def\strmgt{{>}} 

\def\spkovl{{\Cap}} 
\def\spkprt{{\Cup}} 
\def\pkovl{\binap{\spkovl}}
\def\pkprt{\binap{\spkprt}}

\def\setalg{\mathcal}

\def\algtrm#1{\llbracket#1\rrbracket}
\def\salglhs{\setalg{Lhs}} 

\def\algap#1#2{\op{#2}{#1}} 

\def\apos{p}
\def\bpos{q}

\def\aipos{\ip{\apos}} 

\def\aPos{P} 
\def\bPos{Q}

\def\aiPos{\ip{\aPos}} 

\def\spossco{{\cdot}}
\def\possco#1#2{#1\spossco#2}

\def\posemp{\varepsilon}

\def\posone{1}

\def\sposeq{{=}}
\def\poseq{\binap{\sposeq}}

\def\ordpfx#1{#1_{\mathit{o}}} 
\def\spospfx{\ordpfx{\prec}}
\def\spospfxref{\ordpfx{\preceq}}

\def\pospfx{\relap{\spospfx}}
\def\pospfxref{\relap{\spospfxref}}

\def\actx{C} 
\def\bctx{D}
\def\cctx{E}
\def\ctxap#1#2{#1[#2]} 
\def\ctxa{\ctxap{\actx}}

\def\acls{{\varsigma}} 
\def\bcls{{\zeta}}
\def\ccls{{\xi}}

\def\aicls{\ip{\acls}}
\def\bicls{\ip{\bcls}}

\def\sclsle{{\sqsubseteq}} 
\def\clsle{\relap{\sclsle}}

\def\sclsge{{\sqsupseteq}} 
\def\clsge{\relap{\sclsge}}
\def\sclsgt{{\sqsupset}} 
\def\clsgt{\relap{\sclsgt}}

\def\sclslub{{\sqcup}}
\def\clslub{\binap{\sclslub}}
\def\sclsglb{{\sqcap}}
\def\clsglb{\binap{\sclsglb}}

\def\sclsLub{{\bigsqcup}}
\def\clsLub{\ip{\sclsLub}}

\def\sclsin{{\ssetin}}

\def\clsbot{{\bot}} 

\def\clstop{{\top}} 
\def\iclstop{\ip{\clstop}}
\def\sclseq{{\strmeq}} 
\def\clseq{\relap{\sclseq}}
\def\sclsneq{{\strmneq}} 
\def\clsneq{\relap{\sclsneq}}
\def\sclsovl{{\kern-.12em\between\kern-.12em}} 

\def\clsovl{\relap{\sclsovl}}

\def\sclsovlequ{\reltre{\sclsovl}} 
\def\clsovlequ{\relap{\sclsovlequ}}

\def\clslet#1#2#3{{\mathsf{let}\,#1\isdefd#2\,\mathsf{in}\,#3}} 
\def\clsbsz#1{\llfloor#1\rrfloor}
\def\clspsz#1{\llceil#1\rrceil}

\def\Rultrs#1{\mathcal{#1}}
\def\atrs{\Rultrs{\aRul}}
\def\btrs{\Rultrs{\bRul}}

\def\dtrs{\Rultrs{\dRul}}

\def\trsemp{{\setemp}}

\def\subap#1#2{\op{#2}{#1}}
\def\substr#1#2{[#1\mathbin{{:}{=}}#2]}
\def\substra#1#2#3{\subap{\substr{#1}{#2}}{#3}}
\def\asub{{\sigma}}
\def\bsub{{\tau}}

\def\suba{\subap{\asub}}
\def\subb{\subap{\bsub}}

\def\aisub{\ip{\asub}}

\def\msubap#1#2{\subap{#1}{#2}}
\def\msubstr#1#2{\llbracket#1\mathbin{{:}{=}}#2\rrbracket}
\def\msubstra#1#2#3{\msubap{\msubstr{#1}{#2}}{#3}}

\def\subLub{\setLub}

\def\Lab#1{\mathcal{#1}}
\def\aLab{\Lab{L}}
\def\aLabhot{\hot{\aLab}}

\def\Labhota{\funap{\aLabhot}}
\def\alab{{\ell}}
\def\blab{{\kappa}}

\def\hot#1{\mathring{#1}}

\renewcommand{\AA}{\mathcal{A}}
\newcommand{\BB}{\mathcal{B}}
\newcommand{\NN}{\mathbb{N}}
\newcommand{\mul}{{\mathsf{mul}}}
\newcommand{\from}{\leftarrow}
\newcommand{\fromB}[1]{\mathrel{_{#1}{\from}}}

\newcommand{\fromBT}[2]{\mathrel{\prescript{#2}{#1}{\from}}}
\newcommand{\tto}{\twoheadrightarrow}
\newcommand{\tfrom}{\twoheadleftarrow}
\newcommand{\tfromB}[1]{\mathrel{_{#1}{\tfrom}}}
\newcommand{\tfromBT}[2]{\mathrel{\prescript{#2}{#1}{\tfrom}}}

\def\aiarsref{\ip{\aarsref}}
\def\iarsref#1{\relap{\aiarsref}}

\def\aarsinvtre{{\twoheadleftarrow}}
\def\aiarsinvtre{\zp{\aarsinvtre}}

\def\aiarsinvref#1{\lud{#1}{\srelref}{\aarsinv}}

\newcommand{\m}[1]{\mathit{#1}}
\newcommand{\dtrsd}{\dtrs_{\m{d}}}
\newcommand{\atrsp}{\hat{\atrs}}
\newcommand{\dtrsp}{\hat{\dtrs}}
\newcommand{\prsymbol}{\shortmid\!\!\!\shortmid}
\newcommand{\arspr}{\mathrel{\wash{\mbox{}\,\,\,\prsymbol}\mathord{\longrightarrow}}}

\newcommand\ISAFOR{\textsf{Isa\kern-0.2exF\kern-0.2exo\kern-0.2exR}\xspace}
\newcommand\CETA{\textsf{C\kern-0.2exe\kern-0.5exT\kern-0.5exA}\xspace}

\newcommand\SAIGAWA{\textsf{Saigawa}\xspace}

\def\clft#1#2{#1\oplus#2}
\def\rlft#1{#1{\uparrow}}

\def\dctx{F}

\def\bictx{\ip{\bctx}}

\def\src#1{\subap{\sstpsrc}{#1}}

\newcommand{\NH}[1]{#1}

\begin{document}
\title{Confluence by Critical Pair Analysis Revisited
\ifextendedversion\xspace (Extended Version)\fi%
\thanks{Supported by JSPS KAKENHI Grant Number 17K00011 and Core to Core Program.
\ifextendedversion
\NH{This paper is an extended version of \cite{Hiro:Nage:Oost:Oyam:19}.}
\fi
}}
%
%
\author{%
Nao Hirokawa\inst{1}\orcidID{0000-0002-8499-0501}
\and
Julian Nagele\inst{2}\orcidID{0000-0002-4727-4637}
\and
Vincent van Oostrom\inst{3}\orcidID{0000-0002-4818-7383}
\and
Michio Oyamaguchi\inst{4}
}
\authorrunning{N. Hirokawa et al.}
%
\institute{%
JAIST, Japan, \email{hirokawa@jaist.ac.jp}
\and
Queen Mary University of London, UK, \email{j.nagele@qmul.ac.uk}
\and
University of Innsbruck, Austria, \email{Vincent.van-Oostrom@uibk.ac.at}
\and
Nagoya University, Japan, \email{oyamaguchi@za.ztv.ne.jp}
}

\maketitle              
\begin{abstract}
We present two methods for proving confluence of left-linear term rewrite
systems.  One is \emph{hot-decreasingness}, combining the
parallel/development closedness theorems with rule labelling based on a
terminating subsystem.  The other is \emph{critical-pair-closing system},
allowing to boil down the confluence problem to confluence 
of a special subsystem whose duplicating rules are relatively terminating.

\keywords{Term rewriting \and Confluence \and Decreasing diagrams.}
\end{abstract}

\section{Introduction} 
\label{sec:intro}



We present two results for proving confluence of first-order
left-linear term rewrite systems, which extend and generalise
three classical results: Knuth and Bendix' criterion~\cite{Knut:Bend:70} 
and strong and parallel closedness due to Huet~\cite{Huet:80}.
Our idea is to reduce confluence of a term rewrite system $\atrs$
to that of a subsystem $\dtrs$ comprising rewrite rules 
needed for \emph{closing} the critical pairs of $\atrs$.
In Section~\ref{sec:hot} we introduce \emph{hot-decreasingness},
requiring that critical pairs can be closed using rules
that are either below those in the peak or in a
terminating subsystem $\dtrs$.
In Section~\ref{sec:cpcs} we introduce the notion of a 
\emph{critical-pair-closing system} and present a confluence-preservation
result based on relative termination $\dtrsd/\atrs$ of 
the duplicating part $\dtrsd$ of $\dtrs$. 
For the left-linear systems we consider, our first criterion 
generalises both Huet's parallel closedness and Knuth and Bendix' criterion,
and the second Huet's strong closedness.
In Section~\ref{sec:implementation}, we assess viability of the new techniques,
reporting on their implementation and empirical results.

Huet's parallel closedness result relies on the notion
of overlap whose geometric intuition is subtle~\cite{Baad:Nipk:98,Nage:Midd:16}, and
reasoning becomes intricate for development closedness as covered by Theorem~\ref{thm:hot}.
We factor the classical theory of overlaps
and critical pairs through the \emph{encompassment} lattice in which overlapping 
redex-patterns \emph{is} taking their join and the amount of overlap between
redex-patterns is computed \emph{via} their meet,
thus allowing to reason algebraically about overlaps.
Methodologically, our contribution here is the introduction of  
the lattice-theoretic language itself, relevant as it allows one to
reason about occurrences of patterns\footnote{%
Modelled in various ways, via e.g.:
tree homomorphisms (tree automata~\cite{Tata:07}),
term-operations (algebra), 
context-variables, 
labelling (rippling~\cite{Bund:etal:05}), to name a few.
}
and their amount of (non-)overlap, omnipresent in deduction.
Technically, whereas Huet's critical pair lemma~\cite{Huet:80} is well-suited
for proving confluence of \emph{terminating} TRSs, it is ill-suited to do so 
for \emph{orthogonal} TRSs.
Our lattice-theoretic results remedy this, allowing to decompose a reduction $R$ both
\emph{horizontally} (as $\relsco{R_1}{R_2}$) and 
\emph{vertically} (as $\substra{x}{R_2}{R_1}$), enabling both 
termination and orthogonality reasoning in confluence proofs
(Theorem~\ref{thm:hot}).

In the last decade various classical confluence results for \emph{term}
rewrite systems have been factored through the decreasing diagrams
method~\cite{Oost:94,Oost:08} for proving confluence of \emph{abstract} rewrite systems,
often leading to generalisations along the way: 
e.g.\ Felgenhauer's multistep labelling~\cite{Felg:15} generalises 
Okui's simultaneous closedness~\cite{Okui:98},
the layer framework~\cite{Felg:Midd:Zank:Oost:15} 
generalises Toyama's modularity~\cite{Toya:87}, 
critical pair systems~\cite{Hiro:Midd:11} generalise both 
orthogonality~\cite{Rose:73} and Knuth and Bendix' criterion~\cite{Knut:Bend:70},
and Jouannaud and Liu generalise, among others~\cite{Liu:16},
parallel closedness, but in a way we do not know 
how to generalise to development closedness~\cite{Oost:97}.
This paper fits into this line of research.\footnote{%
\ifextendedversion
For space reasons we have omitted the proof by decreasing diagrams of Theorem~\ref{thm:cpcs}
from the main text.  See the appendix for omitted proofs.
\else
For space reasons we have omitted the proof by decreasing diagrams of Theorem~\ref{thm:cpcs}.
\fi
}

We assume the reader is familiar
with term rewriting~\cite{Ders:Joua:90,Baad:Nipk:98,Tere:03} in general
and confluence methods~\cite{Knut:Bend:70,Huet:80,Oost:08} in particular.
Notions not explicitly defined in this paper can all be found in those
works.

\section{Preliminaries on decreasingness and encompassment} 
\label{sec:prelims}

We recall the key ingredients of 
the decreasing diagrams method for proving confluence,
see~\cite{Ohle:02,Tere:03,Oost:08,Liu:16},
and revisit the classical notion of critical pair,
recasting its traditional account~\cite{Knut:Bend:70,Huet:80,Baad:Nipk:98}
based on redexes (substitution instances of left-hand sides)
into one based on redex-patterns (left-hand sides).

\subsubsection{Decreasingness} 
\label{sec:decreasing}

Consider an ARS comprising an $I$-indexed relation
${\to} = \bigcup_{\alab \in I} {\to_\alab}$ equipped with a
well-founded strict order $\succ$.
We refer to 
$\{ \blab \in I \mid \alab \succ \blab \}$ by $\orddwna{\alab}$, and to
$\orddwna{\alab} \cup \orddwna{\blab}$ by $\orddwna{\alab,\blab}$.
For a subset $J$ of $I$ we define $\to_J$ as $\bigcup_{\alab \in J} {\to_\alab}$.
\begin{definition} \label{def:decreasing}
  A diagram for a peak $\iarsinva{\alab}{\bobj}{\iarsa{\blab}{\aobj}{\cobj}}$
  is \emph{decreasing} if its closing conversion has shape
  $\relap{\relsco{\aiarsequ{\orddwna{\alab}}
}{\relsco{\aiarsref{\blab}
}{\relsco{\aiarsequ{\orddwna{\alab,\blab}}
}{\relsco{\aiarsinvref{\alab}
        }{\aiarsequ{\orddwna{\blab}}
        }}}}}{\bobj}{\cobj}$.
  An ARS in this setting is called \emph{decreasing} if every peak
  can be completed into a decreasing diagram.
\end{definition}
One may think of decreasing diagrams as combining the diamond 
property~\cite[Theorem~1]{Newm:42} (via the steps in the closing conversion 
with labels $\alab$, $\blab$) at the basis of confluence of \emph{orthogonal}
systems~\cite{Chur:Ross:36,Rose:73}, with local 
confluence diagrams~\cite[Theorem~3]{Newm:42} (via the conversions with 
labels $\orddwna{\alab,\blab}$) at the basis of 
confluence of \emph{terminating} systems~\cite{Knut:Bend:70,Mayr:Nipk:98}.
\begin{theorem}[\textnormal{\cite{Oost:94,Oost:08}}]  \label{thm:dd:complete}
An ARS is confluent if it is decreasing.
Conversely, every \emph{countable} ARS that is confluent, is decreasing
for some set of indices $I$.
\end{theorem}
For the converse part it suffices that the set of \emph{labels} $I$ is a doubleton, 
a result that can be reformulated without referring to decreasing diagrams, as follows.
\begin{lemma}[\textnormal{\cite{Endr:Klop:Over:18}}] 
  \label{lem:countable:confluent:spanning}
  A countable confluent rewrite relation has a spanning forest.
\end{lemma}
Here a spanning forest for $\aars$ is a relation $\relle{\bars}{\aars}$
that is \emph{spanning} ($\releq{\barsequ}{\aarsequ}$) and a \emph{forest}, i.e.\ 
deterministic ($\arsinvb{\bobj}{\arsb{\aobj}{\cobj}}$ implies $\objeq{\bobj}{\cobj}$) 
and acyclic.

\subsubsection{Critical peaks revisited} 
\label{sec:critical:peaks}

We introduce \emph{clusters} as the structures obtained \emph{after}
the matching of the left-hand side of a rule in a rewrite step, 
but \emph{before} its replacement by the right-hand side.
When proving the aforementioned results in Sections~\ref{sec:hot}
and \ref{sec:cpcs}, we use them as a tool to analyse overlaps and critical peaks.
To illustrate our notions we use the following running example.
We refer to \cite{Tere:03,Hiro:Midd:11} for the notions of multistep.
\begin{example} \label{exa:cluster:1}
  In the TRS $\atrs$ with $\rula{x} \hastype \rulstr{f(f(x))}{g(x)}$
  the term $\atrm \isdefd f(f(f(f(a))))$  allows the step
  $f(\rula{f(a)}) \hastype \arsa{\atrm}{f(g(f(a)))}$ and 
  multistep $\rula{\rula{a}} \hastype \arsdeva{\atrm}{g(g(a))}$.
\end{example}  
Here $f(\rula{f(a)})$ and $\rula{\rula{a}}$ are so-called proofterms,
\emph{terms} representing \emph{proofs} of rewritability in rewriting logic~\cite{Mese:92,Tere:03}. 
The source of a proofterm can be computed by the $2$nd-order substitution
$\sstpsrc$ of the left-hand side of the rule for the rule symbol\footnote{%
$\sstpsrc$ can be viewed as tree homomorphism~\cite{Tata:07},
or as a term algebra $\funap{\algap{\salglhs}{\arul}}{\vec{\atrm}} \isdefd \substra{\vec{\atrmvar}}{\vec{\atrm}}{\alhs}$.} 
$\subap{\sstpsrc}{f(\rula{f(a)})} \isdefd 
 \trmeq{\msubstra{\arul}{\lambda x.f(f(x))}{f(\rula{f(a)})}
      }{f(f(f(f(a))))
      }$, and, \emph{mutatis mutandis}, the same for the target via $\sstptgt$.
Proofclusters, introduced here, 
abstract from such proofterms by allowing to represent the matching and substitution phases 
of multisteps as well, by means of let-expressions.
\begin{example} \label{exa:cluster:2}
  The multistep in Example~\ref{exa:cluster:1}
  comprises three phases~\cite[Chapter~4]{Oost:94}:
  \begin{enumerate}
  \item
    $\clslet{\amtrmvar,\bmtrmvar\hspace{-0.52ex}}{\hspace{-0.52ex}\lambda x.f(f(x)),\lambda y.f(f(y))}{\mtrmvara{\mtrmvarb{a}}}$
    denotes \emph{matching} $f(f(x))$ twice;
  \item
    $\clslet{\amtrmvar,\bmtrmvar\hspace{-0.52ex}}{\hspace{-0.52ex}\lambda x.\rula{x},\lambda x.\rula{x}}{\mtrmvara{\mtrmvarb{a}}}$
    denotes \emph{replacing} by $\arul$ twice;
  \item
    $\clslet{\amtrmvar,\bmtrmvar\hspace{-0.52ex}}{\hspace{-0.52ex}\lambda x.g(x),\lambda x.g(x)}{\mtrmvara{\mtrmvarb{a}}}$
    denotes \emph{substituting} $g(x)$ twice.
  \end{enumerate}
\end{example}
To represent these we assume to have
\emph{proofterms} $\atrm, \btrm, \ctrm, \ldots$ over a signature comprising
\emph{function} symbols $\asym, \bsym, \csym, \ldots$, 
\emph{rule} symbols $\arul, \brul, \crul, \ldots$,
$2$nd-order variables $\amtrmvar, \bmtrmvar, \cmtrmvar, \ldots$,
all having natural number arities, 
and $1$st-order variables $\atrmvar, \btrmvar, \ctrmvar, \ldots$
(with arity $\natzer$).
We call proofterms without $2$nd-order variables or rule symbols, 
\emph{$1$st-order} proofterms respectively \emph{terms}, 
ranged over by $\amtrm$, $\bmtrm$, $\cmtrm$, $\ldots$.
\begin{definition}
  A \emph{proofcluster} is a let-expression 
  $\clslet{\vec{\amtrmvar}}{\vec{Q}}{\atrm}$, where 
  \begin{itemize}
  \item
    $\vec{\amtrmvar}$ is a
    vector $\aimtrmvar{\natone},\ldots,\aimtrmvar{\anat}$
    of (pairwise distinct) second-order variables;
  \item
    $\vec{Q}$ is a vector of length $\anat$
    of closed $\lambda$-terms 
    $\ip{Q}{\aidx} = \lambda \vec{\aitrmvar{\aidx}}.\bitrm{\aidx}$,
    where $\bitrm{\aidx}$ is a proofterm
    and the length of the vector $\vec{\aitrmvar{\aidx}}$ of variables is the arity of $\aimtrmvar{\aidx}$; and
  \item
    $\atrm$ is a proofterm, the \emph{body}, with its $2$nd-order variables among $\vec{\amtrmvar}$.
  \end{itemize}
  Its \emph{denotation} $\algtrm{\clslet{\vec{\amtrmvar}}{\vec{Q}}{\atrm}}$
  is $\msubstra{\vec{\amtrmvar}}{\vec{Q}}{\atrm}$.
  It is a \emph{cluster} if $\bitrm{\natone},\ldots,\bitrm{\anat},\atrm$ are terms.
\end{definition}
We let $\acls$, $\bcls$, $\ccls$, $\ldots$ range over (proof)clusters.
They denote (proof)terms.
\begin{example} \label{exa:cluster:3}
  Using $\acls$, $\bcls$, $\ccls$ for the three 
  let-expressions in Example~\ref{exa:cluster:2}, each
  is a proofcluster and $\acls$, $\ccls$ are clusters.  
  Their denotations are the
  term $\algtrm{\acls} \isdefd \trmeq{f(f(f(f(a))))}{\atrm}$,
  proofterm $\algtrm{\bcls} \isdefd \rula{\rula{a}}$, and
  term $\algtrm{\ccls} \isdefd g(g(a))$.
\end{example}
We assume the usual variable renaming conventions, both for the $2$nd-order
ones in let-binders and the $1$st-order ones in $\lambda$-abstractions.
We say a proofcluster $\acls$ is \emph{linear} if every (let or $\lambda$) binding 
binds exactly once, and~\emph{canonical}~\cite{Meti:83}
if, when a binding variable occurs to the left of another such (of the same type),
then the first bound occurrence of the former occurs before that of the latter 
in the pre-order walk of the relevant proofterm. 
\begin{example} \label{exa:cluster:4}
  Let $\bcls'$ and $\ccls'$ be the clusters
  $\clslet{\amtrmvar}{\lambda x.f(f(x))}{\mtrmvara{\mtrmvara{a}}}$ and
  $\clslet{\amtrmvar,\bmtrmvar}{\lambda yz.f(f(y)),\lambda x.f(f(x))}{\mtrmvarb{\mtrmvara{a,f(a)}}}$.
  Each of $\acls$, $\bcls'$, $\ccls'$ denotes $\atrm$ in Example~\ref{exa:cluster:1}.
  The cluster $\acls$ is linear and canonical,
  $\bcls'$ is canonical but not linear ($\amtrmvar$ occurs twice in the body), and
  $\ccls'$ is neither linear ($z$ does not occur in $f(y)$) 
  nor canonical ($\bmtrmvar$ occurs outside of $\amtrmvar$ in the body).
\end{example}
We adopt the convention that absent $\lambda$-binders are inserted linearly, canonically;
$\clslet{\amtrmvar}{f(f(x))}{\mtrmvara{\mtrmvara{a}}}$ is $\bcls'$.
Clusters witness encompassment $\strmencref$~\cite{Ders:Joua:90}.
\begin{proposition} \label{prop:encompassment}
  $\trmencref{\atrm}{\btrm}$ iff 
  $\exists \ctrm,\amtrmvar$ s.t.\ 
  $\trmeq{\algtrm{\clslet{\amtrmvar}{\btrm}{\ctrm}}}{\atrm}$
  and $\amtrmvar$ occurs once in $\ctrm$.
\end{proposition}
We define the \emph{size} $\trmsiz{\atrm}$ of a proofterm $\atrm$
in a way that is compatible with encompassment.
Formally, $\trmsiz{\atrm}$ is the pair comprising the number 
of non-$1$st-order-variable symbols in $\atrm$, and the sum over the $1$st-order variables $\atrmvar$,
of the \emph{square} of the number of occurrences of $\atrmvar$ in $\atrm$.
Then $\natgt{\trmsiz{\atrm}}{\trmsiz{\btrm}}$ if $\trmenc{\atrm}{\btrm}$, 
where we (ab)use $\snatgt$ to denote the lexicographic 
product of the greater-than relation with itself,
e.g. $\natgt{\trmsiz{g(a,a)} \isdefd \pair{\natthr}{\natzer}
    }{\natgt{\trmsiz{g(x,x)} \isdefd \pair{\natone}{\natfou}
           }{\trmsiz{g(x,y)} \isdefd \pair{\natone}{\nattwo}
           }}$.
For a proofcluster $\acls$ given by
$\clslet{\vec{\atrmvar}}{\vec{\btrm}}{\atrm}$ 
its \emph{pattern}-size $\clspsz{\acls}$ is $\natSum{\aidx}{\trmsiz{\bitrm{\aidx}}}$
(adding component-wise, with empty sum $\pair{\natzer}{\natzer}$) 
and its \emph{body}-size $\clsbsz{\acls}$ is $\trmsiz{\atrm}$.
Encompassment $\strmencref$ is at the basis of the theory of reducibility~\cite[Section~3.4.2]{Tata:07}: 
$\atrm$ is reducible by a rule $\rulstr{\alhs}{\arhs}$ iff $\trmencref{\atrm}{\alhs}$.
For instance, 
$\clslet{\amtrmvar}{f(f(x))}{f(\mtrmvara{f(a)})}$ is a witness 
to reducibility of $\atrm$ in Example~\ref{exa:cluster:1}.
We call it, or simply $f(f(x))$, a pattern in $\atrm$.
\begin{definition} \label{def:pattern}
  Let $\acls$ be a canonical linear proofcluster $\clslet{\vec{\amtrmvar}}{\vec{\btrm}}{\atrm}$
  with term $\atrm$.
  We say $\acls$ is a \emph{multipattern} if each $\bitrm{\aidx}$ is 
  a non-variable $1$st-order term, and 
  $\acls$ is a \emph{multistep} if each $\bitrm{\aidx}$
  has shape $\rula{\vec{\atrmvar}}$, i.e.\ a rule symbol applied to 
  a sequence of pairwise distinct variables.
  If $\vec{\amtrmvar}$ has length $\natone$ we drop the prefix `multi'.
\end{definition}
We use $\aStp,\bStp,\cStp,\ldots$ to range over multisteps, and $\astp,\bstp,\cstp,\ldots$
to range over steps. Taking their denotation yields the usual multistep~\cite{Tere:03,Hiro:Midd:11} 
and step ARSs $\aarsdev$ and~$\aars$ underlying a TRS $\atrs$.
These can be alternatively obtained by first applying $\sstpsrc$ and $\sstptgt$
(of which only the former is guaranteed to yield a multipattern, by left-linearity)
and then taking denotations:
$\trmeq{\algtrm{\subap{\sstpsrc}{\aStp}}
      }{\subap{\sstpsrc}{\algtrm{\aStp}}
      }$ and
$\trmeq{\algtrm{\subap{\sstptgt}{\aStp}}
      }{\subap{\sstptgt}{\algtrm{\aStp}}
      }$.
Pattern- and body-sizes of multipatterns are compositional.
\begin{proposition} \label{prop:multipattern:size}
  For multipatterns $\acls$,$\vec{\acls}$ if
  $\trmeq{\acls}{\substra{\vec{\atrmvar}}{\vec{\acls}}{\aicls{\natzer}}}$
  with each variable among $\vec{\atrmvar}$ occurring once in the body of $\aicls{\natzer}$, then
  $\nateq{\clspsz{\acls}}{\natSum{\aidx}{\aicls{\aidx}}}$, and
  $\natge{\clsbsz{\acls}}{\clsbsz{\aicls{\aidx}}}$ for all $\aidx$,
  with strict inequality holding in case the substitution is not a bijective renaming.
  Here multipattern-substitution substitutes in the body and combines let-bindings.
\end{proposition}
Multipatterns are ordered by refinement $\sclsle$.
\begin{definition} \label{def:refinement}
  Let $\acls$ and $\bcls$ be multipatterns
  $\clslet{\vec{\amtrmvar}}{\vec{\btrm}}{\atrm}$ and 
  $\clslet{\vec{\bmtrmvar}}{\vec{\ctrm}}{\dtrm}$.
  We say $\acls$ \emph{refines} $\bcls$
  and write $\clsle{\acls}{\bcls}$,
  if there is a $2$nd-order substitution $\asub$ on $\vec{\bmtrmvar}$ 
  with $\trmeq{\suba{\dtrm}}{\atrm}$ and
  $\trmeq{\algtrm{\clslet{\vec{\amtrmvar}
                        }{\vec{\btrm}
                        }{\suba{\imtrmvarb{\aidx}{\vec{\bitrmvar{\aidx}}}}}
                        }
        }{\citrm{\aidx}
        }$
  for all $\aidx$,
  with $\vec{\bitrmvar{\aidx}}$ the variables of $\citrm{\aidx}$.
\end{definition}
\begin{example} \label{exa:cluster:6}
  We have $\clsle{\acls}{\acls'}$ with $\acls'$ being $\clslet{\cmtrmvar}{f(f(f(f(z))))}{\mtrmvarc{a}}$,
  and $\acls$ as in Example~\ref{exa:cluster:3},
  as witnessed by the $2$nd-order substitution mapping $\cmtrmvar$ to $\lambda \atrmvar.\mtrmvara{\mtrmvarb{\atrmvar}}$.
\end{example}
\begin{lemma} \label{lem:lattice}
  $\sclsle$ is a finite distributive lattice~\cite{Dave:Prie:90} on
  multipatterns denoting a $1$st-order term $\atrm$, with least element $\clsbot$ the empty let-expression
  $\clslet{}{}{\atrm}$, and greatest element $\clstop$ of shape
  $\clslet{\amtrmvar}{\atrm'}{\mtrmvara{\vec{\atrmvar}}}$ with
  $\vec{\atrmvar}$ the vector of variables in $\atrm$.
\end{lemma}
\begin{proof}[Idea]
  Although showing that $\sclsle$ is reflexive and transitive is easy,
  showing anti-symmetry or existence of/constructions 
  for meets $\sclsglb$ and joins $\sclslub$, directly is not. 
  Instead, it \emph{is} easy to see that each multipattern 
  $\clslet{\vec{\amtrmvar}}{\vec{\btrm}}{\atrm}$ is determined by the set
  of the (non-empty, convex,\footnote{%
Here convex means that for each pair of positions $\apos$,$\bpos$ in
the set, all positions on the shortest path from $\apos$ to $\bpos$
in the term tree are also in the set, cf.~\cite[Definition~8.6.21]{Tere:03}.
}
   pairwise disjoint) sets of node 
  positions of its patterns $\bitrm{\aidx}$ in $\atrm$, and vice versa.
  For instance, the multipatterns 
  $\acls$ and $\acls'$ in Example~\ref{exa:cluster:6} are determined by 
  $\setstr{\setstr{\posemp,\posone},
           \setstr{\possco{\posone}{\posone},\possco{\posone}{\possco{\posone}{\posone}}}}$ and
  $\setstr{\setstr{\posemp,\posone,\possco{\posone}{\posone},
                   \possco{\posone}{\possco{\posone}{\posone}}
                  }}$. 
  Viewing multipatterns as sets in that way
  $\clsle{\acls}{\bcls}$ iff
  $\forall\setin{\aPos}{\acls}$, $\exists\setin{\bPos}{\bcls}$ with $\setle{\aPos}{\bPos}$.
  Saying $\setin{\aPos,\bPos}{\setlub{\acls}{\bcls}}$ 
  have \emph{overlap} if $\setneq{\setglb{\aPos}{\bPos}}{\setemp}$, 
  denoted by $\clsovl{\aPos}{\bPos}$, characterising meets and joins now also is easy:
  $\clsglb{\acls}{\bcls} \isdefd
   \setabs{\setglb{\aPos}{\bPos}
         }{\text{$\setin{\aPos}{\acls}$, $\setin{\bPos}{\bcls}$, and $\clsovl{\aPos}{\bPos}$}}$, and
  $\clslub{\acls}{\bcls} \isdefd 
   \setabs{\setLub{}{\aiPos{\sclsovl}}
         }{\setin{\aPos}{\setlub{\acls}{\bcls}}
         }$, where 
  $\aiPos{\sclsovl} \isdefd \setabs{\setin{\bPos}{\setlub{\acls}{\bcls}}}{\clsovlequ{\aPos}{\bPos}}$,
  i.e.\ the sets connected to $\aPos$ by successive overlaps.
  On this set-representation $\sclsle$ can be shown to be a finite distributive lattice
  by set-theoretic reasoning,
  using that the intersection of two overlapping patterns is a pattern again\footnote{%
This fails for, e.g., connected graphs; these may fall apart into non-connected ones.
}.
  For instance, $\clsbot$ \emph{is} the empty set and
  $\clstop$ \emph{is} the singleton containing the set of all non-variable positions in $\atrm$. 
  \qed
\end{proof}
The (proof of the) lemma allows to freely switch between viewing multisteps and multipatterns as 
let-expressions and as sets of sets of positions,
and to reason about (non-)overlap of multipatterns and multisteps in lattice-theoretic terms.
We show any multistep $\aStp$ can be 
decomposed \emph{horizontally} 
as $\astp$ followed by $\stpres{\aStp}{\astp}$ 
for any step $\stpin{\astp}{\aStp}$~\cite{Hiro:Midd:11,Oost:97},
and \emph{vertically} as some vector $\vec{\aStp}$
substituted in a prefix $\aiStp{\idxout}$ of $\aStp$,
and that peaks can be decomposed correspondingly.
\begin{definition} \label{def:critical}
  For a pair of multipatterns $\acls$,$\bcls$ denoting the same term
  its \emph{amount} of overlap\footnote{%
For the amount of overlap for \emph{redexes} in parallel reduction $\arspr$,
see e.g.~\cite{Huet:80,Baad:Nipk:98,Nage:Midd:16}.}
   and non-overlap is
  $\pkovl{\acls}{\bcls} \isdefd \clspsz{\clsglb{\acls}{\bcls}}$
  respectively 
  $\pkprt{\acls}{\bcls} \isdefd \clsbsz{\clslub{\acls}{\bcls}}$,
  we say $\acls$,$\bcls$ is \emph{overlapping} if 
  $\clsneq{\clsglb{\acls}{\bcls}}{\clsbot}$,
  and \emph{critically} overlapping if moreover
  $\clseq{\clslub{\acls}{\bcls}}{\clstop}$
  and $\trmeq{\algtrm{\acls}}{\algtrm{\bcls}}$ is linear.
  This extends to peaks $\iarsdevinva{\aStp}{\btrm}{\iarsdeva{\bStp}{\atrm}{\ctrm}}$ via 
  $\subap{\sstpsrc}{\aStp}$ and $\subap{\sstpsrc}{\bStp}$.
\end{definition}
Note $\acls$,$\bcls$ is overlapping iff $\natneq{\pkovl{\acls}{\bcls}}{\pair{\natzer}{\natzer}}$.
Critical peaks $\iarsinva{\astp}{\btrm}{\iarsa{\bstp}{\atrm}{\ctrm}}$ 
are classified by comparing the root-positions $\aipos{\astp}$, $\aipos{\bstp}$
of their patterns with respect to the prefix order $\spospfx$,
into being \emph{outer--inner} ($\pospfx{\aipos{\astp}}{\aipos{\bstp}}$),
\emph{inner--outer} ($\pospfx{\aipos{\bstp}}{\aipos{\astp}}$), or 
\emph{overlay} ($\poseq{\aipos{\bstp}}{\aipos{\astp}}$),
and induce the usual~\cite{Knut:Bend:70,Huet:80,Ders:Joua:90,Baad:Nipk:98,Ohle:02,Tere:03}
notion of critical \emph{pair} $\pair{\btrm}{\ctrm}$.\footnote{%
We exclude neither overlays of a rule with itself
nor pairs obtained by symmetry.}
\begin{definition} \label{def:inner}
  A pair $\pair{\acls'}{\bcls'}$ of overlapping patterns such that $\acls'$, $\bcls'$ are
  in the multipatterns $\acls$, $\bcls$ with $\clseq{\clstop}{\clslub{\acls}{\bcls}}$, 
  is called \emph{inner}, if it is minimal among all such pairs, 
  comparing them in the 
  lexicographic product of $\spospfx$ with itself,
  via the root-positions of their patterns, ordering these themselves first by $\spospfxref$. 
  This extends to pairs of steps in peaks of multisteps via $\sstpsrc$.
\end{definition}
\begin{proposition} \label{prop:inner:overlay}
  If $\pair{\astp}{\bstp}$ is an inner pair for a critical peak 
  $\relsco{\aiarsdevinv{\aStp}}{\aiarsdev{\bStp}}$,
  and $\stpin{\astp}{\aStp}$, $\stpin{\bstp}{\bStp}$ contract
  redexes at the same position, 
  then $\clseq{\astp}{\aStp}$ and $\clseq{\bstp}{\bStp}$.     
\end{proposition}
For patterns and peaks of ordinary steps, their join
being top, entails they are overlapping, and 
the patterns in a join are joins of their constituent patterns.
\begin{proposition} \label{prop:critical:peak}
  Linear patterns $\acls$,$\bcls$ are critically overlapping iff 
  $\clseq{\clslub{\acls}{\bcls}}{\clstop}$.
\end{proposition}
\begin{lemma} \label{lem:join:join}
  If 
  $\ccls \isdefd \clslub{\acls}{\bcls}$ 
  and $\clsle{\acls,\bcls}{\ccls}$ are witnessed by the $2$nd-order substitutions $\asub$, $\bsub$,
  for multipatterns $\acls$ and $\bcls$ given by
  $\clslet{\vec{\amtrmvar}}{\vec{\atrm}}{\amtrm}$ and
  $\clslet{\vec{\bmtrmvar}}{\vec{\btrm}}{\bmtrm}$,
  then for all let-bindings $\cmtrmvar\,=\,\ctrm$ of $\ccls$,
  $\trmeq{\iclstop{\ctrm}
        }{\clslub{(\clslet{\vec{\amtrmvar}}{\vec{\atrm}}{\suba{\mtrmvarc{\vec{\ctrmvar}}}})
                }{(\clslet{\vec{\bmtrmvar}}{\vec{\btrm}}{\subb{\mtrmvarc{\vec{\ctrmvar}}}})
                }
        }$.
\end{lemma}
\begin{lemma}[Vertical] \label{lem:vertical}
  A  peak 
  $\iarsdevinva{\aStp
              }{\btrm
    }{\iarsdeva{\bStp
              }{\atrm
              }{\ctrm
              }}$ 
  of overlapping multisteps either is critical or 
  it can be vertically decomposed as:    
  \[ \iarsdevinva{\substra{\vec{\atrmvar}}{\vec{\aStp}}{\aiStp{\idxout}}
                }{\substra{\vec{\atrmvar}}{\vec{\btrm}}{\bitrm{\idxout}}
      }{\iarsdeva{\substra{\vec{\atrmvar}}{\vec{\bStp}}{\biStp{\idxout}}
                }{\substra{\vec{\atrmvar}}{\vec{\atrm}}{\aitrm{\idxout}}
                }{\substra{\vec{\atrmvar}}{\vec{\ctrm}}{\citrm{\idxout}}
                }} \]
  for peaks
  $\iarsdevinva{\aiStp{\aidx}
              }{\bitrm{\aidx}
    }{\iarsdeva{\biStp{\aidx}
              }{\aitrm{\aidx}
              }{\citrm{\aidx}
              }}$ 
  with 
  $\natge{\pkovl{\aStp}{\bStp}
        }{\pkovl{\aiStp{\aidx}}{\biStp{\aidx}}
        }$ and
  $\natgt{\pkprt{\aStp}{\bStp}
        }{\pkprt{\aiStp{\aidx}}{\biStp{\aidx}}
        }$, for all $\aidx$.
\end{lemma}
Let $\aStp$, $\bStp$ in 
$\iarsdevinva{\aStp}{\btrm}{\iarsdeva{\bStp}{\atrm}{\ctrm}}$
be given by 
$\clslet{\vec{\amtrmvar}\!}{\!\vec{\rula{\vec{\atrmvar}}}}{\amtrm}$ and
$\clslet{\vec{\bmtrmvar}\!}{\!\vec{\rulb{\vec{\btrmvar}}}}{\bmtrm}$,
for rules 
$\irula{\aidx}{\vec{\aitrmvar{\aidx}}} \hastype \rulstr{\ailhs{\aidx}}{\airhs{\aidx}}$ and
$\irulb{\bidx}{\vec{\bitrmvar{\bidx}}} \hastype \rulstr{\bilhs{\bidx}}{\birhs{\bidx}}$.
Lemma~\ref{lem:lattice} entails that if $\aStp$, $\bStp$ are non-overlapping
their patterns are (pairwise) disjoint,
so that the join $\clslub{\subap{\sstpsrc}{\aStp}}{\subap{\sstpsrc}{\bStp}}$ 
is given by taking the (disjoint) union of the let-bindings:
$\clslet{\vec{\amtrmvar}\vec{\bmtrmvar}
         }{\vec{\alhs}\vec{\blhs}
         }{\cmtrm
         }$ for some $\cmtrm$ such that
$\trmeq{\msubstra{\vec{\bmtrmvar}}{\vec{\blhs}}{\cmtrm}}{\amtrm}$ and
$\trmeq{\msubstra{\vec{\amtrmvar}}{\vec{\alhs}}{\cmtrm}}{\bmtrm}$.
We define the \emph{join}\footnote{%
This does not create ambiguity with joins of multipatterns since if
$\clsneq{\aStp}{\bStp}$, then $\clsneq{\algtrm{\aStp}}{\algtrm{\bStp}}$
unless the let-bindings of both are empty, so both are bottom.
}
$\clslub{\aStp}{\bStp}$ and \emph{residual} $\stpres{\aStp}{\bStp}$ by
$\clslet{\vec{\amtrmvar}\vec{\bmtrmvar}
       }{\vec{\rula{\vec{\atrmvar}}}\vec{\rulb{\vec{\btrmvar}}}
       }{\cmtrm
       }$ respectively
$\clslet{\vec{\amtrmvar}
       }{\vec{\rula{\vec{\atrmvar}}}
       }{\msubstra{\vec{\bmtrmvar}}{\vec{\brhs}}{\cmtrm}
       }$, where, as substituting the right-hand sides $\vec{\brhs}$
may lose being linear and canonical, we implicitly canonise and 
linearise the latter by reordering and replicating
let-bindings.
Then
$\relap{\relsco{\aiarsdev{\clslub{\aStp}{\bStp}}
              }{\aiarsdevinv{\stpres{\aStp}{\bStp}}
              }
      }{\atrm
      }{\ctrm
      }$, giving rise to the classical
residual theory~\cite{Chur:Ross:36,Huet:Levy:91a,Boud:85,Bare:85},
see~\cite[Section~8.7]{Tere:03}.
We let $\stpin{\astp}{\aStp}$ abbreviate
$\exists \bStp.\clseq{\aStp}{\clslub{\astp}{\bStp}}$.
\begin{example} \label{exa:cluster:5}
  The steps $\astp$ and $\bstp$ given by
  $\clslet{\amtrmvar}{\lambda x.\rula{x}}{\mtrmvara{f(f(a))}}$ respectively
  $\clslet{\amtrmvar}{\lambda x.\rula{x}}{f(f(\mtrmvarb{a}))}$,
  are non-overlapping, $\stpin{\astp,\bstp}{\bcls}$, 
  $\clseq{\clslub{\astp}{\bstp}}{\bcls}$,
  and $\iarsdeva{\stpres{\astp}{\bstp}}{f(f(g(a)))}{g(g(a))}$,
  for $\bcls$ and $\arul$ as in Example~\ref{exa:cluster:3}.
\end{example}  
\begin{lemma}[Horizontal] \label{lem:horizontal}
  A peak 
  $\relap{\relsco{\aiarsdevinv{\aStp}}{\aiarsdev{\bStp}}
        }{\atrm
        }{\btrm
        }$ 
  of multisteps either
  \begin{enumerate}
  \item \label{ite:horizontal:non-overlap}
    is non-overlapping and then
    $\relap{\relsco{\aiarsdev{\stpres{\bStp}{\aStp}}}{\aiarsdevinv{\stpres{\aStp}{\bStp}}}
        }{\atrm
        }{\btrm
        }$, with the rule symbols occurring in 
    $\stpres{\bStp}{\aStp}$ contained in $\bStp$ 
    (and those in $\stpres{\aStp}{\bStp}$ contained in $\aStp$); or
  \item \label{ite:horizontal:overlap}
    it can be horizontally decomposed:
    $\relap{\relsco{\relsco{\aiarsdevinv{\stpres{\aStp}{\astp}}}{\aiarsinv{\astp}}
                  }{\relsco{\aiars{\bstp}}{\aiarsdev{\stpres{\bStp}{\bstp}}}
                  }
          }{\atrm
          }{\btrm
          }$ 
    for some peak $\relsco{\aiarsinv{\astp}}{\aiars{\bstp}}$ 
    of overlapping steps $\stpin{\astp}{\aStp}$ and $\stpin{\bstp}{\bStp}$.
  \end{enumerate} 
\end{lemma}   
The above allows to refactor the proof of the critical pair 
lemma~\cite[Lemma~3.1]{Huet:80} for left-linear TRSs,
as an induction on the amount of \emph{non-overlap} between
the steps in the peak, such that the critical peaks form the \emph{base} case:
\begin{lemma} \label{lem:critical:pair}
  A left-linear TRS is locally confluent if all critical pairs are joinable.
\end{lemma}
\begin{proof}
  We show every peak $\relsco{\aiarsinvref{\astp}}{\aiarsref{\bstp}}$
  of empty or single steps is joinable, 
  by induction on the amount of non-overlap 
  ($\pkprt{\astp}{\bstp}$) ordered by $\strmgt$.
  We distinguish cases on whether $\astp$, $\bstp$ are overlapping 
  ($\natneq{\pkovl{\astp}{\bstp}}{\pair{\natzer}{\natzer}}$) or not.
  If $\astp$, $\bstp$ do not have overlap, in particular when either $\astp$ or $\bstp$ 
  is empty, then we conclude by Lemma~\ref{lem:horizontal}(\ref{ite:horizontal:non-overlap}).
  If $\astp$, $\bstp$ do have overlap, then by Lemma~\ref{lem:vertical} the peak either
  \begin{itemize}
  \item
    is critical and we conclude by assumption; or
  \item
    can be (vertically) \emph{decomposed} into smaller such peaks
    $\relsco{\aiarsinvref{\aistp{\aidx}}
           }{\aiarsref{\bistp{\aidx}}
           }$.
    Since these are $\strmgt$-smaller, the induction hypothesis yields
    them joinable, from which we conclude by reductions and joins being
    closed under \emph{composition}. \qed
  \end{itemize}
\end{proof}
\begin{remark}
  Apart from enabling our proof of Theorem~\ref{thm:hot} below,
  we think this refactoring is methodologically interesting,
  as it extends to (parallel and) simultaneous critical pairs,
  then yielding, we claim, simple statements and proofs 
  of confluence results~\cite{Okui:98,Felg:15} based on these
  and their higher-order generalisations.
\end{remark}

\section{Confluence by hot-decreasingness} 
\label{sec:hot}

Linear TRSs have a critical-pair criterion for so-called 
rule-labelling~\cite{Oost:08,Hiro:Midd:11,Zank:Felg:Midd:15}:
If all critical peaks are decreasing with respect some rule-labelling, 
then the TRS is decreasing, hence confluent.
We introduce the hot-labelling extending that result to left-linear 
TRSs. To deal with non-right-linear rules we make use of a 
rule-labelling for multisteps that is invariant under duplication, 
cf.~\cite{Felg:15,Zank:Felg:Midd:15}.
\begin{remark}
  Na{\"\i}ve extensions fail.
  Non-left-linear TRSs need not 
  be confluent even without critical pairs~\cite[Exercise~2.7.20]{Tere:03}.
  That non-right-linear TRSs need not be confluent even if all critical peaks
  are decreasing for rule-labelling, is witnessed by~\cite[Example~8]{Hiro:Midd:11}.
\end{remark}
\begin{definition}
  For a TRS $\atrs$, terminating subsystem $\dtrs \subseteq \atrs$,
  and labelling of $\setdif{\atrs}{\dtrs}$-rules
  into a well-founded order $\aord$,
  \emph{hot}-labelling $\aLabhot$ maps a multistep 
  $\aStp \hastype \iarsdeva{\atrs}{\atrm}{\btrm}$ 
  \begin{itemize}
  \item
    to the \emph{term} $\atrm$ if $\aStp$ contains $\dtrs$-rules only; and
  \item
    to the \emph{set} of $\aord$-maximal $\setdif{\atrs}{\dtrs}$-rules in $\aStp$ otherwise.
  \end{itemize}
  The \emph{hot}-order $\aordhot$ relates terms by $\aiarstra{\dtrs}$,
  sets by $\aordmul$, and all sets to all terms.
\end{definition}
Note $\aordhot$ is a well-founded order as series composition~\cite{Bech:Groo:Reto:97}
of $\aiarstra{\dtrs}$ and $\aordmul$, which are well-founded orders by the 
assumptions on $\dtrs$ and $\aord$.
Taking the \emph{set} of maximal rules in a multistep
makes hot-labelling invariant under duplication.
As with the notation $\orddwna{\alab}$, we denote
$\{ \blab \mid \ordhota{\alab}{\blab} \}$ by $\ordhotdwna{\alab}$, and
$\{ \blab \mid \ordhotrefa{\alab}{\blab} \}$ by $\ordhotrefdwna{\alab}$.
\begin{definition} \label{def:hot:decreasing:trs}
  A TRS $\atrs$ is \emph{hot-decreasing} if its critical peaks
  are decreasing for the hot-labelling, for some $\dtrs$
  and $\aord$, such that each outer--inner critical peak $\relsco{\aiarsinv{\alab}}{\aars}$ for label $\alab$, 
  is decreasing by a conversion of shape (oi): 
  $\relsco{\aiarsequ{\ordhotdwna{\alab}}}{\aiarsdevinv{\ordhotrefdwna{\alab}}}$.
\end{definition}
\begin{theorem} \label{thm:hot}
  A left-linear TRS is confluent, if it is hot-decreasing.
\end{theorem}
Before proving Theorem~\ref{thm:hot}, we give (non-)examples and special cases.
\begin{example} \label{exa:nats:hot} 
  Consider the left-linear TRS $\atrs$:
  \[\begin{array}{lrcl@{\,\,\,\,}lrcl@{\,\,\,\,}lrcl}
  \arul_1\hastype & \m{nats}                  & \srulstr & \m{0} : \m{inc}(\m{nats}) &
  \arul_3\hastype & \m{inc}(x : y)            & \srulstr & \m{s}(x) : \m{inc}(y)     &
  \arul_5\hastype & \m{hd}(x : y)             & \srulstr & x \\
  \arul_2\hastype & \m{d}(x)                  & \srulstr & x : (x : \m{d}(x))        &
  \arul_4\hastype & \m{inc}(\m{tl}(\m{nats})) & \srulstr & \m{tl}(\m{inc}(\m{nats})) &
  \arul_6\hastype & \m{tl}(x : y)             & \srulstr & y
  \end{array}\]
  By taking $\dtrs \isdefd \trsemp$, labelling rules by themselves, and
  ordering $\orda{\airul{4}}{\airul{1},\airul{3},\airul{6}}$ 
  the only critical peak 
  $\relsco{\aiarsinv{\setstr{\airul{4}}}}{\aiars{\setstr{\airul{1}}}}$
  can be completed into the decreasing diagram:
  \begin{center}
  \begin{tikzpicture}
  \node (t) at (0,0.9) {$\m{tl}(\m{inc}(\m{nats}))$};
  \node (o) at (4,0.9) {$\m{inc}(\m{tl}(\m{nats}))$};
  \node (s) at (8,0.9) {$\m{inc}(\m{tl}(\m{0} : \m{inc}(\m{nats})))$};
  \node (u) at (0,0)   {$\m{tl}(\m{inc}(\m{0} : \m{inc}(\m{nats})))$};
  \node (v) at (4.3,0) {$\m{tl}(\m{s}(\m{0}) : \m{inc}(\m{inc}(\m{nats}))$};
  \node (w) at (8,0)   {$\m{inc}(\m{inc}(\m{nats}))$};
  \draw[->]
    (o) edge node[anchor=north] {$\scriptstyle\{\arul_4\}$} (t)
    (o) edge node[anchor=north] {$\scriptstyle\{\arul_1\}$} (s)
    (t) edge node[anchor=east]  {$\scriptstyle\{\arul_1\}$} (u)
    (u) edge node[anchor=north] {$\scriptstyle\{\arul_3\}$} (v)
    (v) edge node[anchor=north] {$\scriptstyle\{\arul_6\}$} (w)
    (s) edge node[anchor=west]  {$\scriptstyle\{\arul_6\}$} (w)
  ;
  \end{tikzpicture}\vspace{-0.7em}
  \end{center}
  Since the peak is outer--inner, the closing conversion must be of (oi)-shape
  $\relsco{\aiarsequ{\ordhotdwna{\setstr{\airul{4}}}}}{\aiarsdevinv{\ordhotrefdwna{\setstr{\airul{4}}}}}$.
  It is, so
  the system is confluent by Theorem~\ref{thm:hot}.
\end{example}
\begin{example} \label{exa:half:levy:hot}
  Consider the left-linear confluent TRS $\atrs$:
  \[\begin{array}{lrcl@{\quad\quad}lrcl@{\quad\quad}lrcl}
    \arul_1 \hastype & f(a,a) & \srulstr & b      &
    \arul_3 \hastype & f(c,x) & \srulstr & f(x,x) & 
    \arul_5 \hastype & f(c,c) & \srulstr & f(a,c) \\
    \arul_2 \hastype & a      & \srulstr & c      &
    \arul_4 \hastype & f(x,c) & \srulstr & f(x,x)
  \end{array}\]
  Since $b$ is an $\atrs$-normal form,
  the only way to join the outer--inner critical peak
  $\iarsinva{\airul{1}}{b}{\iarsa{\airul{2}}{f(a,a)}{f(c,a)}}$ 
  is by a conversion starting with a step 
  $\iarsinva{\airul{1}}{b}{f(a,a)}$.
  As its label must be identical to the same step in the peak, not smaller,
  whether we choose $\airul{1}$ to be in $\dtrs$ or not,
  the peak is not hot-decreasing, so Theorem~\ref{thm:hot} does not apply.

  That hot-decreasingness in Theorem~\ref{thm:hot} cannot be weakened
  to (ordinary) decreasingness, can be seen by considering $\atrs'$ obtained
  by omitting $\arul_5$ from $\atrs$.  Although $\atrs'$ is not
  confluent~\cite[Example 8]{Hiro:Midd:11},
  by taking $\dtrs = \varnothing$ and $\orda{\airul{1}}{\airul{3},\airul{4}}$,
  we can show that all critical peaks of $\atrs'$ are decreasing for the
  hot-labelling.
\end{example}
A special case of Theorem~\ref{thm:hot}, is that a left-linear terminating 
TRS is confluent~\cite{Knut:Bend:70}, if each critical pair is joinable,
as can be seen by setting $\dtrs \isdefd \atrs$.
\begin{corollary} \label{cor:dev:closed}
  A left-linear development closed TRS is confluent~\cite[Cor.~24]{Oost:97}.
\end{corollary}
\begin{proof}
  A TRS is development closed if for every critical pair $\pair{\atrm}{\btrm}$
  such that $\atrm$ is obtained by an outer step, $\arsdevinva{\atrm}{\btrm}$ holds.
  Taking $\dtrs \isdefd \trsemp$ and labelling all rules the same, say by $0$,
  yields that each outer--inner or overlay critical peak is labelled as
  $\relap{\relsco{\aiarsinv{\setstr{0}}}{\aiars{\setstr{0}}}}{\atrm}{\btrm}$, and can
  be completed as $\iarsdevinva{\setstr{0}}{\atrm}{\btrm}$, 
  yielding a hot-decreasing diagram of (oi)-shape.
  We conclude by Theorem~\ref{thm:hot}. \qed
\end{proof}
The proof of Theorem~\ref{thm:hot} uses the following structural properties 
of decreasing diagrams specific to the hot-labelling. The labelling was designed so they hold.
\begin{lemma}  \label{lem:hot:structural}
  \begin{enumerate}
  \item \label{ite:hot:decreasing}
    If the peak $\iarsdevinva{\alab}{\btrm}{\iarsdeva{\blab}{\atrm}{\ctrm}}$
    is hot-decreasing, then it can be completed into a hot-decreasing
    diagram of shape
    $\iarsequa{\ordhotdwna{\alab}
            }{\btrm
  }{\iarsdeva{\blab
            }{\btrm'
  }{\iarsequa{\ordhotdwna{\alab\blab}
            }{\btrm''
}{\iarsdevinva{\alab
            }{\ctrm''
  }{\iarsequa{\ordhotdwna{\blab}}
            }{\ctrm'
            }{\ctrm
            }}}}$
    such that the $1$st-order variables in all terms in the diagram
    are contained in those of $\atrm$.
  \item \label{ite:hot:commute}
    If the multisteps $\aStp$, $\bStp$ in the peak
    $\iarsdevinva{\aStp}{\btrm}{\iarsdeva{\bStp}{\atrm}{\ctrm}}$
    are non-overlapping, then the valley
    $\relap{\relsco{\aiarsdev{\stpres{\bStp}{\aStp}}
        }{\aiarsdevinv{\stpres{\aStp}{\bStp}}
         }
        }{\btrm
        }{\ctrm
        }$ completes it into a hot-decreasing diagram.    
  \item \label{ite:hot:vertical}
  If the peak 
  $\iarsdevinva{
              }{\btrm
    }{\iarsdeva{
              }{\atrm
              }{\ctrm}}$ and vector of peaks
  $\iarsdevinva{
              }{\vec{\btrm}
    }{\iarsdeva{
              }{\vec{\atrm}
              }{\vec{\ctrm}
              }}$ 
  have hot-decreasing diagrams, so does the composition
  $\iarsdevinva{
              }{\substra{\vec{\atrmvar}}{\vec{\btrm}}{\btrm}
    }{\iarsdeva{
              }{\substra{\vec{\atrmvar}}{\vec{\atrm}}{\atrm}
              }{\substra{\vec{\atrmvar}}{\vec{\ctrm}}{\ctrm}
              }}$.
   \end{enumerate}
\end{lemma}
The proof of Theorem~\ref{thm:hot}
refines our refactored proof (see Lemma~\ref{lem:critical:pair}) 
of Huet's critical pair lemma,
by wrapping the induction on the amount of non-overlap ($\spkprt$)
between multisteps,
into an outer induction on their amount of overlap ($\spkovl$).
\begin{proof}[of Theorem~\ref{thm:hot}]
  We show that every peak 
  $\iarsdevinva{\aStp}{\btrm}{\iarsdeva{\bStp}{\atrm}{\ctrm}}$
  of multisteps $\aStp$ and $\bStp$ 
  can be closed into a hot-decreasing diagram,
  by induction on the pair 
  $\pair{\pkovl{\aStp}{\bStp}}{\pkprt{\aStp}{\bStp}}$ 
  ordered by the lexicographic product of $\snatgt$ with itself.
  We distinguish cases on whether or not $\aStp$ and $\bStp$ have overlap.
 
  If $\aStp$ and $\bStp$ do not have overlap, 
  Lemma~\ref{lem:horizontal}(\ref{ite:horizontal:non-overlap})
  yields
  $\relap{\relsco{\aiarsdev{\stpres{\bStp}{\aStp}}
        }{\aiarsdevinv{\stpres{\aStp}{\bStp}}
         }
        }{\btrm
        }{\ctrm
        }$.
  This valley completes the peak into a hot-decreasing diagram 
  by Lemma~\ref{lem:hot:structural}(\ref{ite:hot:commute}).
  
  If $\aStp$ and $\bStp$ do have overlap, then we further distinguish cases on whether
  or not the overlap is critical.

  If the overlap is not critical, then by Lemma~\ref{lem:vertical}
  the peak can be vertically decomposed into a number of peaks between
  multisteps $\aiStp{\aidx}$, $\biStp{\aidx}$ that have an
  amount of overlap that is not greater,
  $\natge{\pkovl{\aStp}{\bStp}
        }{\pkovl{\aiStp{\aidx}}{\biStp{\aidx}}
        }$, and a strictly smaller amount of non-overlap 
  $\natgt{\pkprt{\aStp}{\bStp}
        }{\pkprt{\aiStp{\aidx}}{\biStp{\aidx}}
        }$.
  Hence the I.H.\ applies and yields that each such peak can be completed into a hot-decreasing
  diagram. We conclude by vertically recomposing them
  yielding a hot-decreasing diagram by Lemma~\ref{lem:hot:structural}(\ref{ite:hot:vertical}).
 
  If the overlap is critical, then by Lemma~\ref{lem:horizontal} the peak can be
  horizontally decomposed as
     $\iarsdevinva{\stpres{\aStp}{\astp}
               }{\btrm
     }{\iarsinva{\astp
               }{\btrm'
        }{\iarsa{\bstp
               }{\atrm
     }{\iarsdeva{\stpres{\bStp}{\bstp}
               }{\ctrm'
               }{\ctrm
               }}}}$ 
  for some peak $\iarsinva{\astp}{\btrm'}{\iarsa{\bstp}{\atrm}{\ctrm'}}$ 
  of overlapping steps $\stpin{\astp}{\aStp}$ and $\stpin{\bstp}{\bStp}$,
  i.e.\ such that $\clseq{\aStp}{\clslub{\astp}{\aStp'}}$ 
  $\clseq{\bStp}{\clslub{\bstp}{\bStp'}}$ for some $\aStp'$, $\bStp'$.
  We choose $\pair{\astp}{\bstp}$ to be \emph{inner} among such overlapping pairs
  (see Definition~\ref{def:inner}),
  assuming w.l.o.g.\ that $\pospfxref{\aipos{\astp}}{\aipos{\bstp}}$
  for the root-positions $\aipos{\astp}$,$\aipos{\bstp}$ of their patterns.
  We distinguish cases on whether or not $\aipos{\astp}$ is a strict prefix of $\aipos{\bstp}$.
   
  If $\poseq{\aipos{\astp}}{\aipos{\bstp}}$, 
  then $\clseq{\astp}{\aStp}$ and $\clseq{\bstp}{\bStp}$ 
  by Proposition~\ref{prop:inner:overlay},
  so the peak \emph{is} overlay, from which we conclude
  since such peaks are hot-decreasing by assumption.
  
  Suppose $\pospfx{\aipos{\astp}}{\aipos{\bstp}}$.
  We will construct a hot-decreasing diagram $D$ for the peak 
  $\iarsdevinva{\aStp}{\btrm}{\iarsdeva{\bStp}{\atrm}{\ctrm}}$
  out of several smaller such diagrams
  as illustrated in Figure~\ref{fig:outer:inner}, 
  using the multipattern
  $\acls \isdefd \clslub{\subap{\sstpsrc}{\aStp}}{\subap{\sstpsrc}{\bstp}}$
  as a basic building block;
  it has as patterns those of $\aStp'$ and the join of the patterns of $\astp$, $\bstp$.
  To make $\acls$ explicit, unfold
  $\aStp$ and $\bStp$ to let-expressions
  $\clslet{\vec{\amtrmvar}}{\vec{\rula{\vec{\atrmvar}}}}{\amtrm}$ respectively
  $\clslet{\vec{\bmtrmvar}}{\vec{\rulb{\vec{\btrmvar}}}}{\bmtrm}$,
  for rules of shapes 
  $\irula{\aidx}{\vec{\aitrmvar{\aidx}}} \hastype \rulstr{\ailhs{\aidx}}{\airhs{\aidx}}$ and
  $\irulb{\bidx}{\vec{\bitrmvar{\bidx}}} \hastype \rulstr{\bilhs{\bidx}}{\birhs{\bidx}}$.
  We let $\vec{\amtrmvar} = \vec{\amtrmvar}'\amtrmvar$
  and $\vec{\bmtrmvar} = \vec{\bmtrmvar}'\bmtrmvar$ be
  such that $\amtrmvar$ and $\bmtrmvar$ are the $2$nd-order variables  
  corresponding to $\stpin{\astp}{\aStp}$ and $\stpin{\bstp}{\bStp}$
  for rules 
  $\rula{\vec{\atrmvar}} \hastype \rulstr{\alhs}{\arhs}$ and
  $\rulb{\vec{\btrmvar}} \hastype \rulstr{\blhs}{\brhs}$.
  By the choice of $\pair{\astp}{\bstp}$ as inner,
  $\subap{\sstpsrc}{\astp}$ is the \emph{unique} pattern in $\subap{\sstpsrc}{\aStp}$ overlapping
  $\subap{\sstpsrc}{\bstp}$.
  As a consequence we can write $\acls$ as
  $\clslet{\vec{\amtrmvar}'\cmtrmvar
         }{\vec{\alhs}'\hat{\atrm}
         }{\cmtrm
         }$, for some pattern $\hat{\atrm}$, the 
  join of the patterns of $\astp$,$\bstp$, such that
  $\asub$ maps $\cmtrmvar$ to a term of shape $\mtrmvara{\vec{\bilhs{\bstp}}}$ 
  as $\astp$ is the outer step,
  and $\bsub$ maps it to a term of shape $\ctxa{\mtrmvarb{\vec{\ailhs{\astp}}}}$,\footnote{%
$\actx$ is a prefix of the left-hand side $\alhs$ of $\arul$. 
For instance, for a peak from $f(g(a))$ between
$\arul \hastype \rulstr{f(g(a))}{\ldots}$ and 
$\brul \hastype \rulstr{g(x)}{\ldots}$, 
$\cmtrmvar$ is mapped by $\asub$ to $\mtrmvara{}$ and 
by $\bsub$ to $f(\mtrmvarb{a})$.}
  where $\asub$, $\bsub$ witness $\clsle{\subap{\sstpsrc}{\aStp},\subap{\sstpsrc}{\bstp}}{\acls}$.  
  That the other $2$nd-order variables are $\vec{\amtrmvar}'$ 
  follows by $\asub$ being the identity on them (their patterns do not overlap $\bstp$),
  and that these are bound to the patterns $\vec{\alhs}'$ by 
  $\bsub$ mapping them to $1$st-order terms (only $\cmtrmvar$ can be mapped to a non-$1$st-order term).
   \begin{figure}[t]
     \def\afig{$\aStp$}
     \def\bfig{$\astp$}
     \def\cfig{$\stpres{\aStp}{\astp}$}
     \def\dfig{$\bStp$}
     \def\efig{$\bstp$}
     \def\ffig{}
     \def\gfig{$\rlft{\hat{\bStp}}$}
     \def\hfig{}
     \def\ifig{}
     \def\jfig{$\clft{\aStp'}{\hat{\aStp}}$}
     \def\kfig{$\stpres{\bStp}{\bstp}$}
     \def\lfig{$\hat{\bStp'}$}
     \def\mfig{$\hat{\aStp'}$}
     \def\nfig{$$}    
     \def\stfig{$\ast$}
     \def\dafig{}
     \def\dbfig{$D'$}
     \def\ihfig{$D_{\text{IH}}$}
      \begin{center}
      \scalebox{0.45}{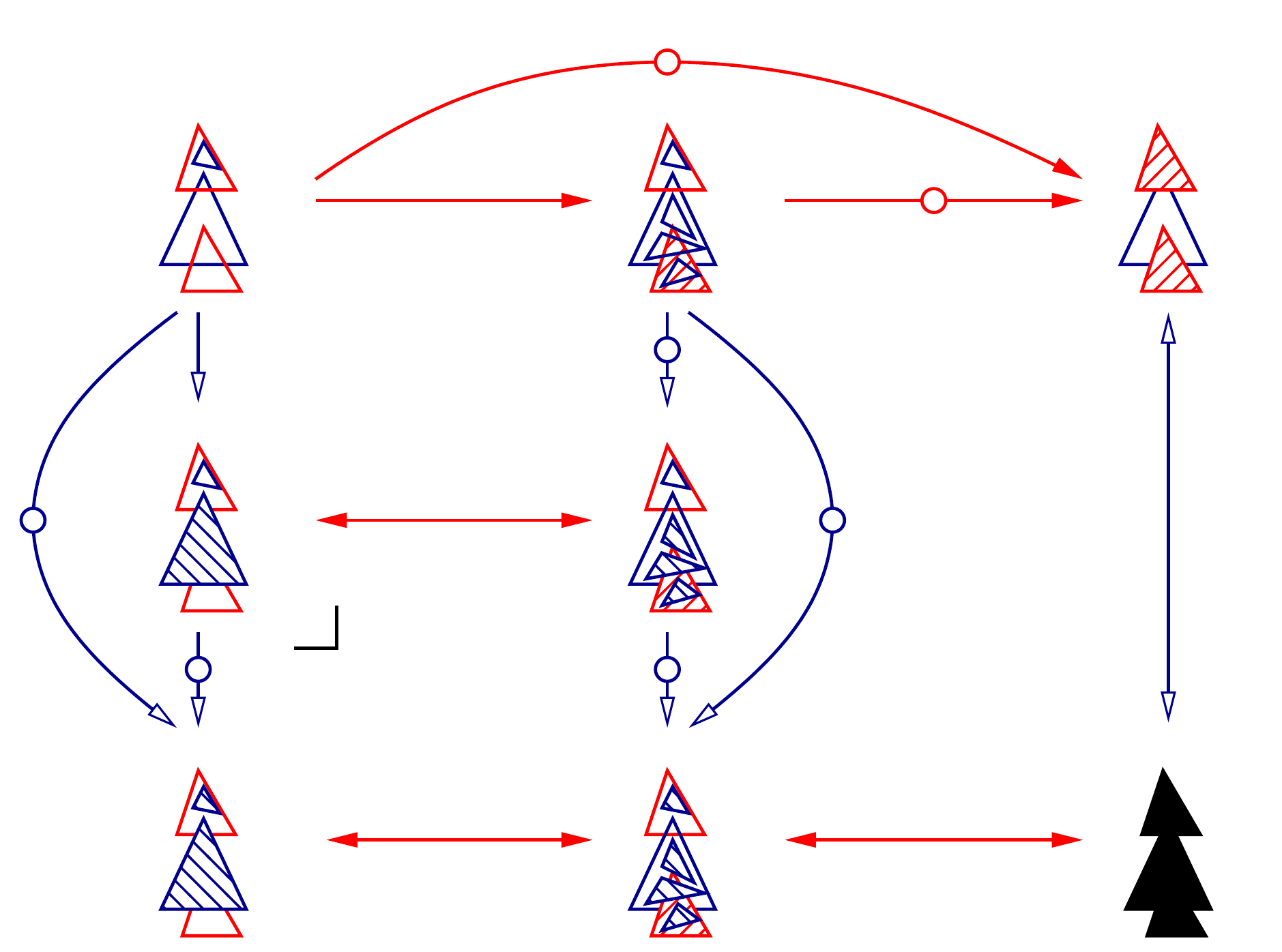}
      \end{center}
       \caption{Outer--inner critical peak construction} 
          \label{fig:outer:inner}
      \end{figure}    
  
  We start with constructing a hot-decreasing diagram $\hot{D}$ 
  for the \emph{critical} peak 
  $\iarsinva{\hat{\astp}}{\hat{\btrm}}{\iarsa{\hat{\bstp}}{\hat{\atrm}}{\hat{\ctrm}}}$ 
  encompassed by the peak between $\astp$ and $\bstp$, as follows.
  We set $\hat{\astp}$ and $\hat{\bstp}$ to 
  $\clslet{\amtrmvar}{\rula{\vec{\atrmvar}}}{\suba{\mtrmvarc{\vec{\ctrmvar}}}}$ respectively
  $\clslet{\bmtrmvar}{\rulb{\vec{\btrmvar}}}{\subb{\mtrmvarc{\vec{\ctrmvar}}}}$.
  This yields a peak as desired, which 
  is outer--inner as $\pospfx{\aipos{\hat{\astp}}}{\aipos{\hat{\bstp}}}$
   by $\pospfx{\aipos{\astp}}{\aipos{\bstp}}$,
  and critical by Lemma~\ref{lem:join:join},
  hence by the hot-decreasingness assumption, 
  it can be completed into a hot-decreasing diagram $\hat{D}$ by a conversion of (oi)-shape:
     $\iarsequa{\ordhotdwna{\Labhota{\hat{\astp}}}
              }{\hat{\btrm}
 }{\iarsdevinva{\ordhotrefdwna{\Labhota{\hat{\astp}}}
              }{\hat{\dtrm}
              }{\hat{\ctrm}
              }}$.
  Below we  refer to its conversion and multistep as $\hat{\bStp}$ and $\hat{\aStp}$.
Based on $\hat{D}$ we construct a hot-decreasing 
    diagram $D'$ (Figure~\ref{fig:outer:inner}, left) for the peak 
    $\iarsdevinva{\aStp
                }{\btrm
         }{\iarsa{\bstp
               }{\atrm
               }{\ctrm'
               }}$
    by constructing a conversion 
    $\rlft{\hat{\bStp}} \hastype
     \arsequa{\btrm
            }{\dtrm''
            }$ and 
    a multistep 
    $\clft{\aStp'}{\hat{\aStp}} \hastype 
     \arsdeva{\ctrm'}{\dtrm''}$,
    with their composition (reversing the latter) of (oi)-shape.
        
    The conversion $\rlft{\hat{\bStp}} \hastype
     \arsequa{\btrm
            }{\dtrm''
            }$
    is constructed by
      lifting the closing conversion $\hat{\bStp}$ of the diagram $\hat{D}$ back into  $\acls$.
    Formally, for any multistep 
    $\hat{\cStp}$ 
    given by
    $\clslet{\vec{\hat{\cmtrmvar}}
           }{\vec{\rulc{\vec{\dtrmvar}}}
           }{\hat{\cmtrm}
           }$ for rules $\irulc{\cidx}{\vec{\ditrmvar{\cidx}}}$,
   occurring anywhere in $\hat{\bStp}$ ,
   we define its \emph{lifting} $\rlft{\hat{\cStp}}$ to be
   $\clslet{\vec{\hat{\cmtrmvar}}
           }{\vec{\rulc{\vec{\dtrmvar}}}
           }{\substra{\vec{\amtrmvar}',\cmtrmvar}{\vec{\arhs}',\hat{\cmtrm}}{\cmtrm}
           }$.
    That is, we update $\acls$
    by substituting\footnote{%
For this to be a valid $2$nd-order substitution, the $1$st-order variables of $\hat{\cStp}$ ($\hat{\cmtrm}$)
must be contained in those of $\hat{\atrm}$, which we may assume by
Lemma~\ref{lem:hot:structural}(\ref{ite:hot:decreasing}).}
   both $\hat{\cStp}$ (for $\cmtrmvar$, instead of binding that to $\hat{\atrm}$)
   and the \emph{right}-hand sides $\vec{\arhs}'$ in its body.
   Because right-hand sides $\vec{\arhs}$ need not be linear, 
   the resulting proofclusters may have to be linearised (by replicating let-bindings) first to obtain multisteps.
   This extends to terms $\etrm$ by $\rlft{\etrm} \isdefd \algtrm{\rlft{(\clslet{}{}{\etrm})}}$.
   That this yields multisteps and terms that connect into a conversion  
        $\iarsequa{\rlft{\hat{\bStp}}
              }{\trmeq{\btrm}{\rlft{\hat{\btrm}}}
              }{\trmeq{\rlft{\hat{\dtrm}}}{\dtrm''}
              }$
   as desired follows by \emph{computation}.
   E.g., 
      $\clseq{\btrm
     }{\clseq{\substra{\vec{\amtrmvar}',\amtrmvar}{\vec{\arhs}',\arhs}{\amtrm}
     }{\clseq{\substra{\vec{\amtrmvar}',\cmtrmvar}{\vec{\arhs}',\hat{\btrm}}{\cmtrm}
            }{\rlft{\hat{\btrm}}
            }}}$
    using that $\asub$ witnesses $\clsle{\subap{\sstpsrc}{\aStp}}{\acls}$ so that
    $\clseq{\amtrm}{\suba{\cmtrm}}$  
    and
    $\trmeq{\hat{\btrm}
          }{\algtrm{\clslet{\amtrmvar}{\arhs}{\suba{\mtrmvarc{\vec{\ctrmvar}}}}}
          }$.
    That the labels in $\rlft{\hat{\bStp}}$ are strictly below $\Labhota{\aStp}$
    follows for \emph{set}-labels from that lifting clearly does not introduce rule symbols
    and from that labels of rule symbols in $\hat{\bStp}$ are, by assumption, strictly below the 
    label of the rule $\arul$ of $\astp$. In case $\aStp$ is \emph{term}-labelled, by $\atrm$, it follows
    from closure of $\aiars{\dtrs}$-reduction under lifting (which also contracts $\aStp'$).   
    
  The multistep $\clft{\aStp'}{\hat{\aStp}} \hastype \arsdeva{\ctrm'}{\dtrm''}$
  is the \emph{combination} of the multisteps $\aStp'$ (the redex-patterns in $\aStp$ other than $\astp$) 
  and $\hat{\aStp}$, lifting the latter into $\acls$.
  For $\hat{\aStp} \hastype \arsdeva{\hat{\ctrm}}{\hat{\dtrm}}$ given by
  $\clslet{\vec{\hat{\amtrmvar}}}{\vec{\funap{\hat{\arul}}{\hat{\atrmvar}}}}{\hat{\amtrm}}$,
  it is defined as
  $\clslet{\vec{\amtrmvar}'\vec{\hat{\amtrmvar}}
         }{\vec{\rula{\vec{\atrmvar}}'}\vec{\funap{\hat{\arul}}{\vec{\hat{\atrmvar}}}}
         }{\substra{\cmtrmvar}{\hat{\amtrm}}{\cmtrm}
         }$.
  Per construction it
  only contracts rules in $\aStp'$, $\hat{\aStp}$, 
  so has a label in $\ordhotrefdwna{\Labhota{\aStp}}$
  by $\clseq{\aStp}{\clslub{\astp}{\aStp'}}$ and the 
  label of $\hat{\aStp}$ is in $\ordhotrefdwna{\Labhota{\hat{\astp}}}$ by the (oi)-assumption.
  That $\clft{\aStp'}{\hat{\aStp}} \hastype \arsdeva{\ctrm'}{\dtrm''}$
  follows again by \emph{computation}, e.g.\ 
   $\trmeq{\algtrm{\clslet{\vec{\amtrmvar}'\vec{\hat{\amtrmvar}}
                         }{\vec{\arhs}'\vec{\hat{\arhs}}
                         }{\substra{\cmtrmvar}{\hat{\amtrm}}{\cmtrm}
                         }}
  }{\trmeq{\substra{\vec{\amtrmvar}',\cmtrmvar
                  }{\vec{\arhs}',\substra{\vec{\hat{\amtrmvar}}}{\vec{\hat{\arhs}}}{\hat{\amtrm}}
                  }{\cmtrm
                  }
  }{\trmeq{\substra{\vec{\amtrmvar}',\cmtrmvar}{\vec{\arhs}',\hat{\dtrm}}{\cmtrm}
  }{\trmeq{\rlft{\hat{\dtrm}}
         }{\dtrm''
         }}}}$.

  Finally, applying the I.H.\ to the peak 
  $\iarsdevinva{\clft{\aStp'}{\hat{\aStp}}
              }{\dtrm''
    }{\iarsdeva{\stpres{\bStp}{\bstp}
              }{\ctrm'
              }{\ctrm
              }}$ yields some hot-decreasing diagram $D_{\text{IH}}$ (Figure~\ref{fig:outer:inner}, right).
  Prefixing $\rlft{\hat{\bStp}}$ to its closing conversion between $\dtrm''$ and $\ctrm$,
  then closes the original peak 
    $\iarsdevinva{\aStp
               }{\btrm
     }{\iarsdeva{\bStp
               }{\atrm
               }{\ctrm
               }}$ into a hot-decreasing diagram $D$,
  because 
  labels of steps in $\rlft{\hat{\bStp}}$ are in $\ordhotdwna{\Labhota{\aStp}}$, 
    $\ordhotrefa{\Labhota{\aStp}
               }{\Labhota{\clft{\aStp'}{\hat{\aStp}}}
               }$ as seen above, and
     $\ordhotrefa{\Labhota{\bStp}
                }{\Labhota{\stpres{\bStp}{\bstp}}
                }$.
  The I.H.\ applies since
  $\pkovl{\aStp}{\bStp} > \pkovl{(\clft{\aStp'}{\hat{\aStp}})}{(\stpres{\bStp}{\bstp})}$:
  To see this, we define
  $\cmtrm' \isdefd \substra{\vec{\cmtrmvar}'}{\vec{\alhs'}}{\cmtrm}$ and
  $\dctx' \isdefd \clslet{\vec{\hat{\amtrmvar}}
           }{\vec{\hat{\alhs}}
           }{\substra{\cmtrmvar}{\hat{\amtrm}}{\cmtrm'}
           }$ and collect needed ingredients (the joins are disjoint):\\
  $\begin{array}{l@{\quad=\quad}lclcl}
    \bctx & \subap{\sstpsrc}{\aStp} & \isdefd &
     \clslub{(\clslet{\vec{\amtrmvar'}}{\vec{\alhs}'}{\substra{\cmtrmvar}{\hat{\atrm}}{\cmtrm}}
          )}{\src{\astp}} & = & \clslub{\src{\aStp'}}{\src{\astp}}\\
    \cctx & \subap{\sstpsrc}{\bStp} & \isdefd &
    \clslub{(\clslet{\vec{\bmtrmvar}'}{\vec{\blhs}'}{\substra{\bmtrmvar}{\blhs}{\bmtrm}})
          }{\src{\bstp}
          }
     &
     = & \clslub{\src{\bStp'}}{\src{\bstp}} \\
     \bctx' & \subap{\sstpsrc}{(\clft{\aStp'}{\hat{\aStp}})} & \isdefd &
     \clslub{(\clslet{\vec{\amtrmvar}'
           }{\vec{\alhs}'
           }{\substra{\cmtrmvar}{\hat{\ctrm}}{\cmtrm}
           })}{\dctx'} \\
     \cctx' & \subap{\sstpsrc}{(\stpres{\bStp}{\bstp})} & \isdefd &
     \clslet{\vec{\bmtrmvar}'}{\vec{\blhs}'}{\substra{\bmtrmvar}{\brhs}{\bmtrm}}
   \end{array}$\\  
   Using these one may reason with sets of patterns (not let-expressions
   as $\trmneq{\atrm}{\btrm'}$; the sets are positions in both $\atrm$,$\btrm'$)
   as follows, relying on distributivity:
   \begin{equation} \label{eqn:decrease}
     \clsgt{(\clsglb{\bctx}{\cctx})
    }{\clseq{(\clsglb{\bictx{-}}{\cctx}) 
    }{\clseq{(\clsglb{\bictx{-}}{\cctx'})
    }{\clsge{(\clsglb{\bictx{+}'}{\cctx'}) 
           }{(\clsglb{\bctx'}{\cctx'})
           }}}} 
    \end{equation}
    where $\dctx$ is the singleton
    $\setstr{\setabs{\setin{\apos}{\subap{\sstpsrc}{\astp}}}{\aipos{\bstp} \npreceq \apos}}$
   having all positions in $\astp$ not below $\bstp$'s root,
    $\bictx{-} \isdefd 
     \clslub{\src{\aStp'}
           }{\dctx
           }$, and
    $\bictx{+}' \isdefd \clslub{(\setdif{\bctx'}{\dctx'})}{\dctx}$. \qed
\end{proof}

\section{Confluence by critical-pair closing systems} 
\label{sec:cpcs}

We introduce a confluence criterion based on identifying for
a term rewrite system $\atrs$ a subsystem $\dtrs$ such that
every $\atrs$-\emph{critical peak} can be \emph{closed} by means of
$\dtrs$-conversions,
rendering the rules used in the peak redundant.
\begin{definition} \label{def:cpcs}
A TRS $\dtrs$ is \emph{critical-pair closing} for a TRS
$\atrs$, if $\dtrs$ is a subsystem of $\atrs$ 
(namely $\dtrs \subseteq \atrs$)
and  $\iarsequa{\dtrs}{s}{t}$ holds for all critical
pairs $\pair{s}{t}$ of $\atrs$.
\end{definition}
We phrase the main result of this section as a preservation-of-confluence result.
We write $\to_{\btrs/\atrs}$ for
$\tto_\atrs \cdot \to_\btrs \cdot \tto_\atrs$, and if it is terminating,
$\btrs/\atrs$ is said to be \emph{(relatively) terminating}.
By $\dtrsd$ we denote the set of all duplicating rules in $\dtrs$.
\begin{theorem} \label{thm:cpcs}
  If $\dtrs$ is a critical-pair-closing system for a left-linear TRS $\atrs$
  such that $\dtrsd/\atrs$ is terminating, 
  then $\atrs$ is confluent if $\dtrs$ is confluent.
\end{theorem}
Any left-linear TRS is critical-pair-closing for itself.
However, the power of the method relies on choosing \emph{small} $\dtrs$.
Before proving Theorem~\ref{thm:cpcs}, we illustrate it by some (non-)examples 
and give a special case.
\begin{example} \label{exa:nats} \label{ex:nats:2}
  Consider again the TRS $\atrs$ in Example~\ref{exa:nats:hot}.
  As we observed, the only critical pair originating from $\airul{4}$ and
  $\airul{1}$ is closed by
  $\relsco{\aiars{\airul{1}}
 }{\relsco{\aiars{\airul{3}}
 }{\relsco{\aiars{\airul{6}}
         }{\aiarsinv{\airul{6}}
         }}}$. 
  So the subsystem $\dtrs \isdefd \{ \arul_1, \arul_3, \arul_6 \}$
  is a critical-pair-closing system for $\atrs$.
  As all $\dtrs$-rules are linear,  $\dtrsd/\atrs$ is vacuously terminating.
  Thus, by Theorem~\ref{thm:cpcs} it is sufficient to show confluence of $\dtrs$.  
  Because $\dtrs$ has no critical pairs, the empty TRS
  $\setemp$ is a critical-pair-closing TRS for $\dtrs$.
  As $\setemp/\dtrs$ is terminating, confluence of $\dtrs$ follows from
  that of $\setemp$, which is trivial.
\end{example}
Observe how confluence was shown by successive applications of the theorem.
\begin{remark}
  In our experiments (see Section~\ref{sec:implementation}),
  $\frac{3}{4}$ of the TRSs proven confluent by means
  of Theorem~\ref{thm:cpcs} used more than $1$ iteration, 
  with the maximum number of iterations being $6$.
  For countable ARSs (see Corollary~\ref{cor:cpcs:ars} below) $1$ iteration suffices, 
  which can be seen by setting $\dtrs$ to the spanning 
  forest obtained by Lemma~\ref{lem:countable:confluent:spanning}.
  This provides the intuition underlying rule specialisation 
   in Example~\ref{exa:half:levy:1} below.
\end{remark}
\begin{example} \label{exa:half:levy}
  Although confluent, the TRS $\atrs$ in Example~\ref{exa:half:levy:hot}
  does not have any confluent critical-pair-closing subsystem
  $\dtrs$ such that $\dtrsd/\atrs$ is terminating,
  not even $\atrs$ itself:
  Because of $b$ being in normal form in the critical pair induced by
  $\iarsinva{\arul_1}{b}{\iarsa{\arul_2}{f(a,a)}{f(a,c)}}$,
  any such subsystem must contain $\arul_4$,
  as one easily verifies, but $\arul_4$ is both duplicating and non-terminating (looping).
  
  Note that the termination condition of $\dtrsd/\atrs$ cannot be
  omitted from Theorem~\ref{thm:cpcs}.  Although the TRS $\atrs'$ in
  Example~\ref{exa:half:levy:hot} is not confluent, it admits the confluent
  critical-pair-closing system $\{ \arul_1, \arul_3, \arul_4 \}$.
\end{example}
\begin{remark}
  The example is taken from~\cite{Hiro:Midd:11} where it was used
  to show that decreasingness of \emph{critical} peaks
  need not imply that of \emph{all} peaks, for rule labelling.
  That example, in turn was adapted from L\'evy's TRS in~\cite{Huet:80} 
  showing that strong confluence need not imply confluence for left-linear TRSs.
\end{remark}
\begin{example}
  For \emph{self-joinable} rules, i.e.\ rules that are self-overlapping and 
  whose critical pairs need further applications of the rule itself to join,
  Theorem~\ref{thm:cpcs} is not helpful since 
  the critical-pair-closing system $\dtrs$ then contains the rule itself.
  Examples of self-joinable rules are associativity 
  $\rulstr{(x \cdot y)\cdot z}{x\cdot(y \cdot z)}$ and self-distributivity
  $\rulstr{(x \cdot y)\cdot z}{(x\cdot z)\cdot(y \cdot z)}$, 
  with confluence of the latter being known to be hard (currently no tool
  can handle it automatically).\footnote{See problem 127 of
\url{http://cops.uibk.ac.at/results/?y=2019-full-run&c=TRS}.}
\end{example}
The special case we consider is that of TRSs that \emph{are} ARSs,
i.e.\ where all function symbols are nullary.
The identification is justified by that any ARS 
in the standard sense~\cite{Ohle:02,Tere:03} can be presented as 
$\aiars{\atrs}$ for the TRS $\atrs$ having a nullary symbol for each
object, and a rule for each step of the ARS. 
Since ARSs have no duplicating rules, Theorem~\ref{thm:cpcs} specialises to
the following result.
\begin{corollary} \label{cor:cpcs:ars}
  If $\dtrs$ is critical-pair-closing for ARS $\atrs$,
  $\atrs$ is confluent if $\dtrs$ is.
\end{corollary}
\begin{example} \label{exa:cpcs:ars}
  Consider the TRS $\atrs$ given by 
  $\arsa{\cobj}{\arsa{\aobj'}{\arsa{\aobj}{\bobj}}}$ and
  $\arsa{\aobj}{\arsa{\aobj'}{\cobj}}$.
  It is an ARS having the critical-pair-closing system $\dtrs$ given by the first part
  $\arsa{\cobj}{\arsa{\aobj'}{\arsa{\aobj}{\bobj}}}$.
  Since $\dtrs$ is orthogonal it is confluent by Corollary~\ref{cor:cpcs:ars},
  so $\atrs$ is confluent by the same corollary.
  In general, a confluent ARS may have many non-confluent
  critical-pair-closing systems. Requiring local confluence is no impediment to that:
  The subsystem $\dtrs'$ of $\atrs$ obtained by removing $\arsa{\cobj}{\aobj'}$ 
  allows to join all $\atrs$-critical peaks, but is not confluent;
  it simply is Kleene's example~\cite[Figure~1.2]{Tere:03} showing that 
  local confluence need not imply confluence.
\end{example}
For $\dtrsd/\atrs$ to be vacuously terminating it is sufficient that all
rules are linear.%
\begin{example} \label{ex:running}
  Consider the linear TRS $\atrs$ consisting of
  $\rho_1\colon f(x) \to f(f(x))$,
  $\rho_2\colon f(x) \to g(x)$, and
  $\rho_3\colon g(x) \to f(x)$.
  The subsystem $\dtrs = \{ \rho_1, \rho_3 \}$ is critical-pair-closing 
  and has no critical pairs, so $\atrs$ is confluent.
\end{example}
From the above it is apparent that, whereas
usual redundancy-criteria are based on \emph{rules} 
being redundant, the theorem gives a sufficient criterion 
for \emph{peaks} of steps being redundant. 
This allows one to leverage the power of extant confluence methods.
Here we give a generalisation of Huet's strong closeness
theorem~\cite{Huet:80} as a corollary of Theorem~\ref{thm:cpcs}.
\begin{definition} \label{def:strongly:closed}
  A TRS $\atrs$ is \emph{strongly closed}~\cite{Huet:80} if
  $s \tto_\atrs^{} \cdot \fromBT{\atrs}{=} t$ and
  $s \to_\atrs^= \cdot \tfromBT{\atrs}{} t$ 
  hold
  for all critical pairs
  $\pair{s}{t}$.
\end{definition}
\begin{corollary} \label{cor:gsc}
A left-linear TRS $\atrs$ is confluent if there exists a
critical-pair-closing system $\dtrs$ for $\atrs$ such that $\dtrs$ is
linear and strongly closed.
\end{corollary}
\begin{example} \label{exa:strongly:closed}
  Consider the linear TRS $\atrs$:
  \[\begin{array}{lrcl@{\quad}lrcl}
    \arul_1\hastype & h(f(x,y))    & \srulstr & f(h(r(x)),y) &
    \arul_2\hastype & f(x,k(y,z))  & \srulstr & g(p(y),q(z,x)) \\
    \arul_3\hastype & h(q(x,y))    & \srulstr & q(x,h(r(y))) &
    \arul_4\hastype & q(x,h(r(y))) & \srulstr & h(q(x,y)) \\
    \arul_5\hastype & h(g(x,y))    & \srulstr & g(x,h(y)) \\
    \arul_6\hastype & a(x,y,z)     & \srulstr & h(f(x,k(y,z))) &
    \arul_7\hastype & a(x,y,z)     & \srulstr & g(p(y),q(z,h(r(x)))) 
  \end{array}\]
  $\dtrs = \{ \arul_1,\ldots,\arul_5 \}$ is critical-pair-closing for $\atrs$,
  since the $\atrs$-critical peak between $\arul_6$ and $\arul_7$ can be $\dtrs$-closed:
  $\iarsa{\arul_1
        }{h(f(x,k(y,z)))
 }{\iarsa{\arul_2
        }{f(h(r(x)),k(y,z))
        }{g(p(y),q(z,h(r(x))))
        }}$.
  Because $\dtrs$ is strongly closed and also linear, 
  confluence of $\atrs$ follows by Corollary~\ref{cor:gsc}.
\end{example}
\begin{remark}
  Neither of the TRSs in Examples~\ref{ex:running} and~\ref{exa:strongly:closed}
  is strongly closed. The former not, because
  $f(f(x)) \tto_\atrs \cdot \fromBT{\atrs}{=} g(x)$ does not hold, 
  and the latter not because
  $g(p(y),q(z,h(r(x)))) \tto_\atrs \cdot \fromBT{\atrs}{=} h(f(x,k(y,z)))$ does not hold.
\end{remark}
Having illustrated the usefulness of Theorem~\ref{thm:cpcs}, we now present
its proof.
In TRSs there are \emph{two} types of peaks: overlapping
and non-overlapping ones.
As Example~\ref{exa:half:levy} shows, confluence criteria only addressing
the former need not generalise from ARSs to TRSs. 
Note that one of the peaks showing non-confluence of $\atrs'$, 
the one between $\arul_2$ and $\arul_3$ ($\arul_4$), is non-overlapping.
Therefore, restricting to a subsystem without $\arul_2$ 
can only provide a partial analysis of confluence of $\atrs'$;
the (non-overlapping) interaction between $\dtrs$ and $\setdif{\atrs}{\dtrs}$ 
is not accounted for, and indeed that is fatal here.
The intuition for our proof is that the problem is that
the number of such interactions is unbounded due to the presence
of the duplicating and non-terminating rule $\arul_3$ (and $\arul_4$) in $\dtrs$,
and that requiring termination of $\dtrsd/\atrs$ bounds that number
and suffices to regain confluence.
This is verified by showing that $\tto_\dtrs \cdot \arsdev_\atrs$ has 
the diamond property.
\begin{lemma}
\label{lem:abstract_criterion}
Let ${\to_\AA} = \bigcup_{a \in I} {\to_a}$ be a relation equipped with a
well-founded order $\succ$ on a label set $I$, and let $\to_\BB$ be a
confluent relation with ${\to_\BB} \subseteq {\tto_\AA}$.  The relation
${\to_\AA}$ is confluent if 
\begin{enumerate}
\item
\(\relsco{\aiarsinv{a}}{\aiars{b}}
  \subseteq
  (\relsco{\aiars{\AA}}{\aiarsinv{\AA}}) \cup
  \bigcup_{\{ a,b \} \succ_\mul \{ a',b' \}}
  (\relsco{\aiarsinv{a'}}{\relsco{\aiarsequ{\BB}}{\aiars{b'}}})
\)
for all 
$a, b \in I$; and
\item
\(
{\relsco{\aiarsinv{a}}{\aiars{\BB}}}
\subseteq
 \relsco{(\aiarstre{\BB}}{\aiarsinv{a})} \cup
 \bigcup_{a \succ a'} (\relsco{\aiarstre{\BB}}{\relsco{\aiarsinv{a'}}{\aiarsequ{\BB}}})
\)
for all 
$a \in I$.
\end{enumerate}
Here $\succ_\mul$ stands for the multiset extension of $\succ$. 
\end{lemma}
\begin{proof}[Sketch]
Let ${\rightarrowtail} \isdefd \relsco{\aiarstre{\BB}}{\aiars{\AA}}$.
We claim that
\(
  \relsco{\aiarsinv{a}
}{\relsco{\aiarsnco{m}{\BB}
}{\relsco{\aiarsinvnce{n}{\BB}
        }{\aiars{b}
        }}}
  \subseteq
  \relsco{{\rightarrowtail}
        }{{\leftarrowtail}
        }
\)
holds for all labels $a, b$ and numbers $m, n \geqslant 0$.
The claim is shown by well-founded induction on $(\{a,b\}, m + n)$ with
respect to the lexicographic product of $\succ_\mul$ and the greater-than
order $\snatgt$ on $\NN$. 
Thus, the diamond property of
${\rightarrowtail}$ follows from the claim and confluence of $\BB$.  As
${\to_\AA} \subseteq {\rightarrowtail} \subseteq {\tto_\AA}$, we conclude
confluence of $\AA$ by e.g.~\cite[Proposition~1.1.11]{Tere:03}.
\end{proof}
\begin{proof}[of Theorem~\ref{thm:cpcs} by Lemma~\ref{lem:abstract_criterion} ]
  Let $I$ comprise pairs of a term and a natural number,
  and define $\iarsa{\pair{\atrmp}{\anat}}{\atrm}{\btrm}$ if
  $\iarstrea{\atrs}{\atrmp}{\iarsdeva{\atrs}{\atrm}{\btrm}}$ 
  with $\anat$ the maximal length of a development of the multistep,\footnote{%
By the Finite Developments Theorem lengths of such developments are finite~\cite{Tere:03}.}
  and $\aiars{\BB} \isdefd \aiars{\dtrs}$,
  in Lemma~\ref{lem:abstract_criterion}.
  As well-founded order $\succ$ on indices we
  take the lexicographic product of $\dtrsd/\atrs$ and 
  greater-than $\snatgt$.
  Henceforth the proof and its case analysis of peaks
  follows the above decreasing-diagrams-based proof.
  Because of this, we only present the interesting case, leaving the others to the reader:
  \begin{itemize}
  \item
    Suppose $\iarsinva{\pair{\atrmp}{\anat}}{\btrm}{\iarsa{\dtrs}{\atrm}{\ctrm}}$
    where the steps do not have overlap. Then 
    by Lemma~\ref{lem:horizontal}(\ref{ite:horizontal:non-overlap}),
    $\relap{\relsco{\aiarsdev{\dtrs}}{\aiarsdevinv{\atrs}}}{\btrm}{\ctrm}$,
    so $\relap{\relsco{\aiarstre{\dtrs}}{\aiarsdevinv{\atrs}}}{\btrm}{\ctrm}$.
    Distinguish cases on the type of 
    the $\dtrs$-rule employed in $\iarsa{\dtrs}{\atrm}{\ctrm}$.
    
    If the rule is duplicating, then 
    $\relap{\relsco{\aiarstre{\dtrs}}{\aiarsinv{\pair{\ctrm}{\bnat}}}}{\btrm}{\ctrm}$
    for $\bnat$ the maximal length of a development of the
    $\aiarsdev{\atrs}$-step from $\ctrm$, and condition~2 is satisfied as
    $\iarsa{\dtrsd}{\atrm}{\ctrm}$ implies
    $\pair{\atrmp}{\anat} \succeq \pair{\ctrm}{\bnat}$.
    
    If the rule is non-duplicating, then
    $\relap{\relsco{\aiarstre{\dtrs}}{\aiarsinv{\pair{\atrmp}{\anat}}}}{\btrm}{\ctrm}$
    as $\iarstrea{\atrs}{\atrmp}{\iarsa{\atrs}{\atrm}{\ctrm}}$ by assumption and
    the length of the maximal development of the residual multistep 
    does not increase
    when projecting over a linear rule.
    Again, condition~2 is satisfied. \qed
  \end{itemize}
\end{proof}

\section{Implementation and experiments} 
\label{sec:implementation}

The presented confluence techniques have been implemented in the confluence
tool \SAIGAWA version~1.12~\cite{Hiro:Klei:12}.
We used the tool to test the criteria on $432$ left-linear TRSs in
COPS~\cite{Hiro:Nage:Midd:18}
Nos.~1--1036, where we ruled out duplicated problems.  Out of $432$
systems, $224$ are known to be confluent and $173$ are non-confluent.

We briefly explain how we automated the presented techniques.  As
illustrated in Example~\ref{ex:nats:2}, 
Theorem~\ref{thm:cpcs} can be used as
a stand-alone criterion.
The condition $s \to_\atrs^* \cdot \fromBT{\atrs}{=} t$ of strong
closedness is tested by
$s \to_\atrs^{\leqslant 5} \cdot \fromBT{\atrs}{=} t$.
For a critical peak $s \from \cdot \to t$ of $\dtrs$,
hot-decreasingness is tested by
$s \tto_\dtrs \cdot \tfromB{\dtrs} t$.
For any other critical peak $s \fromB{\alab} \cdot \to t$, 
we test
the disjunction of
\(
s \to_{\ordhotdwna{\alab}}^{\leqslant 5} \cdot
\mathrel{\aiarsdevinv{\ordhotrefdwna{\alab}}} t
\)
and
\(
s \tto_\dtrs \cdot \mathrel{\aiarsdevinv{\ordhotrefdwna{\alab}}} t
\)
if it is outer--inner one, and if it is overlay, the disjunction of
\(
s \to_{\ordhotdwna{\alab}}^{\leqslant 5} \cdot
\mathrel{\aiarsdevinv{\ordhotrefdwna{\alab}}} t
\)
and $s \tto_\dtrs \cdot \tfromB{\dtrs} t$ is used.  Order constraints for
hot-labeling are solved by SMT solver Yices~\cite{Dute:14}.
For proving (relative) termination we employ the termination tool NaTT
version~1.8~\cite{Yama:Kusa:Saka:14}.
Finally, suitable subsystems $\dtrs$ used in our criteria are searched by
enumeration.

Table~\ref{tab:experiments} gives a summary of the
results.\footnote{Detailed data are available from:
\url{http://www.jaist.ac.jp/project/saigawa/19cade/}}
The tests were run on a PC equipped with Intel Core 
i7-8500Y CPU (1.5 GHz) and 16 GB memory using a timeout of $60$ seconds.
For the sake of comparison we also tested
Knuth and Bendix' theorem (\textsf{kb}),
the strong closedness theorem (\textsf{sc}), and
development closedness theorem (\textsf{dc}).
As theoretically expected, 
they are subsumed by their generalizations.
\begin{table}[tb]
\caption{Experimental results}
\label{tab:experiments}
\centering
\begin{tabular}{l@{\quad}c@{\quad}c@{\quad}c@{\quad}c@{\quad}c@{\quad}c}
\toprule
& Thm.~\ref{thm:hot} 
& Thm.~\ref{thm:cpcs}
& Cor.~\ref{cor:gsc}
& \textsf{kb}
& \textsf{dc} 
& \textsf{sc}
\\[0.5ex]
\# proved (\# timeouts)
& 101 (46)
& 81 (24)
& 94 (15)
& 45 (18)
& 34 (1)
& 62 (1)
\\
\bottomrule
\end{tabular}
\end{table}

\section{Conclusion and future work} 
\label{sec:conclusion}

We presented two methods for proving confluence of TRSs,
dubbed critical-pair-closing systems and hot-decreasingness.
We gave a lattice-theoretic characterisation of overlap.
Since many results in term rewriting, and beyond, are 
based on reasoning about overlap, which is notoriously hard~\cite{Nage:Midd:16},
we expect that formalising our characterisation 
could simplify or even enable formalising them.
We expect that both methods 
generalise to commutation, extend to
HRSs~\cite{Mayr:Nipk:98}, and can be 
strengthened by considering rule \emph{specialisations}.
\begin{example} \label{exa:half:levy:1}
  Analysing the TRS $\atrs$ of Example~\ref{exa:half:levy} 
  one observes that for closing the critical pairs
  only (non-duplicating) \emph{instances} of the duplicating rules
  $\arul_3$ and $\arul_4$ are used. Adjoining these specialisations 
  allows the method to proceed:
  Adjoining 
  $\arul_{3}(a) \hastype \rulstr{f(c,a)}{f(a,a)}$ and
  $\arul_{4}(a) \hastype \rulstr{f(a,c)}{f(a,a)}$ to $\atrs$
  yields a (reduction-equivalent) TRS
  having critical-pair-closing system
  $\{ \arul_1, \arul_{3}(a), \arul_{4}(a), \arul_5 \}$.
  Since this is a linear system without critical pairs,
  it is confluent, so $\atrs$ is as well.
\end{example}

%
%
\bibliographystyle{splncs04}
\bibliography{draft}

\ifextendedversion

\appendix

\section{Proofs omitted from or only sketched in the main text}

That having two labels suffices for the converse part of Theorem~\ref{thm:dd:complete},
is an immediate consequence of Lemma~\ref{lem:countable:confluent:spanning}:
give edges on the tree label $0$, others label $1$.
Our proof of the lemma is a reformulation of the proof of~\cite[Lemma~18]{Endr:Klop:Over:18}.
\begin{proof}[of Lemma~\ref{lem:countable:confluent:spanning}]
  Let $\aars$ be countable and confluent. It suffices to show 
  that if $\aars$ has a single connected component, i.e.\ if
  $\aarsequ$ relates all objects, then we can construct a tree 
  $\relle{\bars}{\aars}$ that spans $\aars$ in the sense that
  $\aars$-convertible objects have a common reduct in the tree:
  $\releq{\aarsequ}{\relsco{\reltre{\bars}}{\ep{\barsinv}{\sreltre}}}$
  (and by determinism then have a least common reduct).
  Because of the countable confluence assumption, there is
  a cofinal reduction~\cite{Klop:80}, i.e.\ a reduction 
  $\arsa{\aiobj{\natzer}}{\arsa{\aiobj{\natone}}{\ldots}}$
  such that for all $\aobj$, there exists $\aidx$ with $\arstrea{\aobj}{\aiobj{\aidx}}$.
  By removing from it any repetitions, we may assume 
  the reduction does not contain cycles~\cite[Proposition~2.2.9]{Oost:94}.
  Taking as initial tree $T_0$ this reduction, its trunk, 
  we construct for each $j$, 
  the tree $T_{j+1}$ by adjoining to $T_j$ any
  reduction from the $j$th (in the countable enumeration) object $a^j$ 
  to $T_j$, stopping at the moment we reach $T_j$ (possibly immediately).
  That a reduction from each $a^j$ to $T_j$ exists holds by cofinality of the trunk $T_0$
  (and monotonicity of the process: $\relle{T_j}{T_{j+1}}$ for all $j$).
  To preserve being a tree, we again remove any repetitions
  from the adjoined reduction.
  That this construction is correct, yields a tree, follows from that
  the trunk $T_0$ is a straight line by construction, and that at
  no stage do we lose determinism: we only adjoin (edges from) nodes not yet
  in the tree and do not introduce cycles per construction. \qed
\end{proof}
\begin{remark}
  The trunk-construction in the proof follows~\cite{Klop:80,Oost:94}, 
  and the tree-construction follows~\cite{Endr:Klop:Over:18}. 
  Compared to~\cite{Oost:94,Endr:Klop:Over:18}
  the proof does not use a minimal distance argument, only cycle-removal.
  A question is whether a spanning tree
  can be constructed in a `single pass'.
\end{remark}
\begin{proof}[of Proposition~\ref{prop:encompassment}]
  $\trmeq{\ctxa{\suba{\btrm}}
        }{\msubstra{\amtrmvar
                  }{\lambda \vec{\atrmvar}.\btrm
                  }{\ctxa{\mtrmvara{\vec{\suba{\atrmvar}}}}
                  }
        }$
  with $\vec{\atrmvar}$ the vector of variables in $\atrm$. \qed
\end{proof}

\begin{proposition} \label{prop:multipattern:compose}
  For  multipatterns $\aicls{\aidx}$ and $\bicls{\aidx}$ such that
  $\trmeq{\algtrm{\aicls{\aidx}}}{\algtrm{\bicls{\aidx}}}$ for all $\aidx$,
  $\clsle{\substra{\vec{\atrmvar}
                 }{\vec{\acls}
                 }{\aicls{\natzer}
                 }
        }{\substra{\vec{\atrmvar}
                 }{\vec{\bcls}
                 }{\bicls{\natzer}
                 }
        }$ iff
   $\clsle{\aicls{\aidx}
         }{\bicls{\aidx}
         }$ for all $\aidx$.
\end{proposition}
\begin{proof}
  Let $\acls \isdefd \substra{\vec{\atrmvar}
                 }{\vec{\acls}
                 }{\aicls{\natzer}
                 }$
                 and $\bcls \isdefd
  \substra{\vec{\atrmvar}
                 }{\vec{\bcls}
                 }{\bicls{\natzer}
                 }$.
  If the $2$nd-order substitutions $\aisub{\aidx}$ witness
  $\clsle{\aicls{\aidx}
         }{\bicls{\aidx}
         }$ for all $\aidx$,
  then $\subLub{\aidx}{\aisub{\aidx}}$ witnesses 
  $\clsle{\acls}{\bcls}$ 
  (assuming $2$nd-order variables 
   are renamed apart).
  Conversely, a witnessing substitution $\asub$ for 
  $\clsle{\acls}{\bcls}$, can be decomposed into 
  $\aisub{\aidx}$ as to the $2$nd-order variables in the bodies of the $\bicls{\aidx}$. \qed
\end{proof}

\begin{proof}[of Proposition~\ref{prop:critical:peak}]
  For $\acls$, $\bcls$ \emph{multi}patterns, 
  $\clseq{\clslub{\acls}{\bcls}}{\clstop}$ entails by Lemma~\ref{lem:lattice}
  that all positions in $\clslub{\acls}{\bcls}$ 
  are related via the `has overlap' relation for patterns in $\acls$,$\bcls$.
  However, if also $\clseq{\clsglb{\acls}{\bcls}}{\clsbot}$ then
  all patterns would be disjoint, yielding either
  $\clseq{\acls}{\clstop}$ and $\clseq{\bcls}{\clsbot}$ or vice versa.
  This is impossible for patterns, as these are non-empty. \qed
\end{proof}

\begin{proof}[of Proposition~\ref{prop:inner:overlay}]
  Intuitively, $\astp$, $\bstp$ cannot be overlapped from above by other steps in 
  $\aStp$, $\bStp$ because the root-positions of their contracted redexes
  are the same, and not from below because of $\pair{\astp}{\bstp}$ being inner. 
  Formally, by assumption the root positions $\aipos{\astp}$
  and $\aipos{\bstp}$ of the contracted redexes are the same.
  By the peak between $\aStp$ and $\bStp$ being critical,
  each redex-pattern in $\aStp$ overlaps some redex-pattern in
    $\bStp$ and vice versa, as each pattern in $\subap{\sstpsrc}{\aStp}$,$\subap{\sstpsrc}{\bStp}$ 
    is connected to each other such pattern in the has-overlap-with relation in their join,
    as shown in Lemma~\ref{lem:lattice}.
  Since $\poseq{\aipos{\astp}}{\aipos{\bstp}}$, no pattern $\stpin{\cstp}{\aStp,\bStp}$
  could overlap $\astp$, $\bstp$ from above, i.e.\ has overlap with them
  such that $\pospfxref{\aipos{\cstp}}{\aipos{\astp}}$.
  This means (using as before that terms are trees and that patterns are convex\footnote{%
For term graphs this fails. There, due to non-convexity of patterns/left-hand sides, 
one may have that part of $\cstp$ overlaps $\astp$ from below 
but at the same time $\pospfx{\aipos{\cstp}}{\aipos{\astp}}$.
}) that in fact for every $\cstp$, $\pospfxref{\aipos{\astp}}{\aipos{\cstp}}$.
  But overlapping $\astp$, $\bstp$ strictly from below, i.e.\ such that
  $\pospfx{\poseq{\aipos{\astp}}{\aipos{\bstp}}}{\aipos{\cstp}}$ 
  is also impossible by the choice of $\astp$,$\bstp$ being inner.
  We conclude that $\cstp$ \emph{is} $\astp$ or $\bstp$, since $\aStp$ and $\bStp$ are multisteps
  and the steps in a multistep are pairwise non-overlapping, from which the claim follows. \qed

\end{proof}

\begin{proof}[of Lemma~\ref{lem:join:join}]
  We have $\trmeq{\algtrm{\clslet{\vec{\amtrmvar}}{\vec{\atrm}}{\suba{\mtrmvarc{\vec{\ctrmvar}}}}}
 }{\trmeq{\ctrm
        }{\algtrm{\clslet{\vec{\bmtrmvar}}{\vec{\btrm}}{\subb{\mtrmvarc{\vec{\ctrmvar}}}}}
        }}$ by definition of $\asub$, $\bsub$ being witnesses to $\clsle{\acls,\bcls}{\ccls}$.
  If there were to exist a $\sclsle$-upperbound of 
  $\clslet{\vec{\amtrmvar}}{\vec{\atrm}}{\suba{\mtrmvarc{\vec{\ctrmvar}}}}$ and
  $\clslet{\vec{\bmtrmvar}}{\vec{\btrm}}{\subb{\mtrmvarc{\vec{\ctrmvar}}}}$ smaller than
  the top $\iclstop{\ctrm}$ of the refinement lattice for $\ctrm$,
  say with witnessing substitutions $\asub'$, $\bsub'$,
  this would contradict $\ccls$ being the join of $\acls$, $\bcls$,
  as updating $\asub$ by mapping the $2$nd-order variables in $\mtrmvarc{\vec{\ctrmvar}}$
  according to $\asub'$, and correspondingly updating $\bsub$ by $\bsub'$,
  would witness an $\sclsle$-upperbound of $\acls$,$\bcls$ smaller than $\ccls$. \qed.
\end{proof}

\begin{proof}[of Lemma~\ref{lem:vertical}]
  Let the multisteps $\aStp$ and $\bStp$ be given by 
  $\clslet{\vec{\bmtrmvar}}{\vec{\rula{\vec{\btrmvar}}}}{\bmtrm}$ respectively
  $\clslet{\vec{\cmtrmvar}}{\vec{\rulb{\vec{\ctrmvar}}}}{\cmtrm}$,
  for rules 
  $\irula{\aidx}{\vec{\bitrmvar{\aidx}}} \hastype \rulstr{\ailhs{\aidx}}{\airhs{\aidx}}$ and
  $\irulb{\bidx}{\vec{\citrmvar{\bidx}}} \hastype \rulstr{\bilhs{\bidx}}{\birhs{\bidx}}$.
  Defining
      $\acls \isdefd \subap{\sstpsrc}{\aStp}$,
      $\bcls \isdefd \subap{\sstpsrc}{\bStp}$,
  and $\ccls \isdefd \clslub{\acls}{\bcls}$,
  we have 
  $\clsle{\acls}{\ccls}$ and $\clsle{\bcls}{\ccls}$ 
  are witnessed (see Definition~\ref{def:refinement})
  by some $2$nd-order substitutions 
  $\asub \isdefd \substr{\vec{\amtrmvar}}{\vec{\bmtrm}}$ and 
  $\bsub \isdefd \substr{\vec{\amtrmvar}}{\vec{\cmtrm}}$,
  with $\vec{\amtrmvar}$ the $2$nd-order variables in the body $\amtrm$ of $\ccls$.  
  If the peak is not critical, either $\clsneq{\ccls}{\clstop}$
  or $\atrm$ and hence $\amtrm$ is not linear.
  By the assumption that $\aStp$ and $\bStp$ have overlap,
  $\amtrm$ must (as the meet $\sclsle$-relates to it)
  contain at least one $2$nd-order variable $\amtrmvar$.
  
  If otherwise only $1$st-order variables occur in $\amtrm$,
  then by assumption it must be non-linear, so of
  shape $\mtrmvara{\vec{\dtrmvar}}$ having some
  repeated variable. Then we decompose the peak
  into a linear prefix and a renaming.
  That is, we choose $\vec{\atrmvar}$ to be linear
  of the same length $\anat$ as $\vec{\dtrmvar}$, and define
  $\aiStp{\idxout}$ and $\biStp{\idxout}$ as
  $\clslet{\vec{\bmtrmvar}
         }{\vec{\rula{\vec{\btrmvar}}}
         }{\msubap{\asub}{\mtrmvara{\vec{\atrmvar}}}
         }$ respectively
  $\clslet{\vec{\cmtrmvar}
         }{\vec{\rulb{\vec{\ctrmvar}}}
         }{\msubap{\bsub}{\mtrmvara{\vec{\atrmvar}}}
         }$ (linearisations of $\aStp$ and $\bStp$)
   and $\aiStp{\aidx}$, $\biStp{\aidx}$ both to 
   $\clslet{}{}{\ditrmvar{\aidx}}$
   (simply renaming $\aitrmvar{\aidx}$ into $\ditrmvar{\aidx}$).
 
   Otherwise, we can write $\amtrm$ as 
   $\substra{\aitrmvar{\natone}}{\aimtrm{\idxin}}{\aimtrm{\idxout}}$
   for terms $\aimtrm{\aidx}$ in which at least one non-$1$st-order variable occurs
   (for instance, let $\aimtrm{\idxout}$ be obtained by
    replacing one of the arguments of the head-symbol of $\amtrm$ by $\aitrmvar{\natone}$).
   Then we choose
   the multisteps $\aiStp{\aidx}$ and $\biStp{\aidx}$ to be
    $\clslet{\vec{\bmtrmvar}
           }{\vec{\rula{\vec{\btrmvar}}}
           }{\suba{\aimtrm{\aidx}}
           }$ respectively
    $\clslet{\vec{\cmtrmvar}
           }{\vec{\rulb{\vec{\ctrmvar}}}
           }{\subb{\aimtrm{\aidx}}
           }$ 
     (and canonising the let-bindings, deleting binders for $2$nd-order variables
      occurring in the \emph{other} body, i.e.\ in $\aimtrm{\natdif{\natone}{\aidx}}$).
     Again, the decomposed peaks are smaller in size.
     
     By simple computations one verifies that in both cases the decomposed
     peaks compose to the original peak. 
     Using Propositions~\ref{prop:multipattern:compose} 
     (for compositionality of $\sclslub$ and $\sclsglb$)
     and~\ref{prop:multipattern:size} we compute
  \[ \nateq{\pkovl{\aStp}{\bStp}
   }{\nateq{\clspsz{\clsglb{\substra{\vec{\atrmvar}
                                   }{\subap{\sstpsrc}{\vec{\aStp}}
                                   }{(\subap{\sstpsrc}{\aiStp{\natzer}})
                                   }
                          }{\substra{\vec{\atrmvar}
                                   }{\subap{\sstpsrc}{\vec{\bStp}}
                                   }{(\subap{\sstpsrc}{\biStp{\natzer}})
                                   }
                          }}
   }{\nateq{\natSum{\aidx}{\clspsz{\clsglb{\subap{\sstpsrc}{\aiStp{\aidx}}}{\subap{\sstpsrc}{\biStp{\aidx}}}}}
          }{\natSum{\aidx}{\pkovl{\aiStp{\aidx}}{\biStp{\aidx}}}
          }}} \]
     and similarly 
     $\natge{\pkprt{\aStp}{\bStp}
           }{\pkprt{\aiStp{\aidx}}{\biStp{\aidx}}
           }$ for all $\aidx$.
    In the $1$st-order-variable-only case we conclude strict inequality by
   $\natgt{\nateq{\pkprt{\aStp}{\bStp}
                }{\nateq{\trmsiz{\mtrmvara{\vec{\dtrmvar}}}}{\pair{\natone}{\anat'}}
                }
         }{\nateq{\trmsiz{\mtrmvara{\vec{\atrmvar}}}}{\pair{\natone}{\anat}},
           \nateq{\trmsiz{\ditrmvar{\aidx}}}{\pair{\natzer}{\natone}}
         }$, for some $\natgt{\anat'}{\anat}$.
     In the other case, strict inequality follows by the choice of splitting
     into $\aimtrm{\idxout}$ in $\aimtrm{\idxin}$ in such a way that
     both contain at least one non-$1$st-order variable symbol.
     \qed
\end{proof} 

\begin{proof}[of Lemma~\ref{lem:horizontal}]
  The $1$st item holds per construction of residuals as given above.
  For the $2$nd item, 
  we obtain by Lemma~\ref{lem:lattice},
  that multipatterns are the join of their patterns
  so that we can write
  $\clseq{\aStp}{\clsLub{\stpin{\astp}{\aStp}}{\astp}}$ and
  $\clseq{\bStp}{\clsLub{\stpin{\bstp}{\bStp}}{\bstp}}$,
  hence by distributivity 
  $\clseq{\clsglb{\subap{\sstpsrc}{\aStp}}{\subap{\sstpsrc}{\bStp}}
        }{\clsLub{\stpin{\astp}{\aStp},\stpin{\bstp}{\bStp}
                }{\clsglb{\subap{\sstpsrc}{\astp}}{\subap{\sstpsrc}{\bstp}}
                }
        }$
  so that $\aStp$ and $\bStp$ have overlap iff
  some of their constituting steps have overlap.
  We conclude per construction of residuals. \qed
\end{proof}

\begin{proof}[of Lemma~\ref{lem:abstract_criterion}]
Defining ${\rightarrowtail} \isdefd \relsco{\aiarstre{\BB}}{\aiarstre{\AA}}$,
it suffices (cf.\ e.g.\ \cite[Proposition~1.1.11]{Tere:03}) 
to show ${\rightarrowtail}$ has the diamond property, as 
${\to_\AA} \subseteq {\rightarrowtail} \subseteq {\tto_\AA}$
by the assumption that $\aiarstre{\BB} \subseteq \aiarstre{\AA}$.
We claim 
\[
  \relsco{\aiarsinv{a}
}{\relsco{\aiarsnco{m}{\BB}
}{\relsco{\aiarsinvnce{n}{\BB}
        }{\aiars{b}
        }}}
  \subseteq
  \relsco{{\rightarrowtail}
        }{{\leftarrowtail}
        }
\]
holds for all labels $a, b$ and numbers $m, n \geqslant 0$.  
From the claim and confluence of $\to_\BB$ we conclude since 
\(
  \relsco{{\leftarrowtail}
        }{{\rightarrowtail}
        }
  \subseteq
  \relsco{\aiarsinv{a}
}{\relsco{\aiarstre{\BB}
}{\relsco{\aiarsinvtre{\BB}
        }{\aiars{b}
        }}}
  \subseteq
  \relsco{{\rightarrowtail}
        }{{\leftarrowtail}
        }
\).
The claim is shown by well-founded induction on
$(\{a,b\}, m + n)$ with respect to the lexicographic product of
$\succ_\mul$ and the greater-than order $\snatgt$ on $\NN$.
We distinguish cases, depending on whether or not $\natgt{m + n}{\natzer}$.
If $m = n = 0$ then
\[ 
  \relsco{\aiarsinv{a}
}{\relsco{\aiarsnco{m}{\BB}
}{\relsco{\aiarsinvnce{n}{\BB}
        }{\aiars{b}
        }}}
\subseteq
  (\relsco{\aiars{\AA}}{\aiarsinv{\AA}}) \cup
   \bigcup_{\{ a,b \} \succ_\mul \{ a',b' \}}
  (\relsco{\aiarsinv{a'}}{\relsco{\aiarsequ{\BB}}{\aiars{b'}}})
\subseteq  
{\rightarrowtail \cdot \leftarrowtail}   \]
by condition~1, and by confluence of $\aiars{\BB}$ and the I.H., respectively.
Otherwise, assume w.l.o.g.\ that $m > 0$ and consider
\(
  \relap{\relsco{\aiarsinv{a}}{\aiars{\BB}}
       }{x
}{\relap{\relsco{\aiarsnco{m-1}{\BB}
       }{\relsco{\aiarsinvnce{n}{\BB}
               }{\aiars{b}
               }}
       }{y
       }{z
       }}
\).
By condition~2 applied to $\relsco{\aiarsinv{a}}{\aiars{\BB}}$
and confluence of $\aiars{\BB}$, either 
\[
\text{$
x \tto_\BB \cdot \fromB{a} y \to_\BB^{m-1} \cdot \fromBT{\BB}{n} \cdot \to_b z\quad$
or
$\quad\relap{\relsco{\aiarstre{\BB}
           }{\relsco{\aiarsinv{a'}
           }{\relsco{\relsco{\aiarstre{\BB}}{\aiarsinvtre{\BB}}
                   }{\aiars{b}
                   }}}
        }{x
        }{z
        }$}
\]
for some $a \succ a'$.
In both cases the I.H.\ applies, because of a decrement in the second
component respectively a decrease in the first, 
and we conclude to 
$\relap{\relsco{\aiarstre{\BB}}{\relsco{{\rightarrowtail}}{{\leftarrowtail}}}
      }{x
      }{z
      }$,
hence to 
$\relap{\relsco{{\rightarrowtail}}{{\leftarrowtail}}
      }{x
      }{z
      }$. \qed
\end{proof}

\begin{proof}[of Theorem~\ref{thm:cpcs} by decreasing diagrams]
  Let $\dtrs$ be a critical-pair-closing system for $\atrs$
  such that $\dtrsd/\atrs$ is terminating and $\aiars{\dtrs}$ is countable.
  Let $\atrsp \isdefd \setdif{\atrs}{\dtrs}$ and $\dtrsp \isdefd \setdif{\dtrs}{\dtrsd}$,
  so that $\atrsp, \dtrsp, \dtrsd$ forms a partition of (the rules of) $\atrs$
  and $\dtrsp, \dtrsd$ of $\dtrs$, and consider the following labellings
  of steps in a conversion:
  \begin{itemize}
  \item
    a multistep $\iarsdeva{\atrsp}{\atrm}{\btrm}$ is labeled by a triple
    $\triple{\anat}{\atrmp}{\bnat}$ with $\anat$ denoting the number of
    $\aiarsdevinv{\atrsp}$-steps to its left in the conversion
    (symmetrically, $\aiarsdev{\atrsp}$-steps to its right for $\aiarsdevinv{\atrsp}$-steps),\footnote{%
Our peak-transformations preserve 
these numbers for \emph{other} steps~\cite[multi-labels]{Liu:16}.
} 
    $\atrmp$ a term $\iarstrea{\atrs}{\atrmp}{\atrm}$ (a so-called 
     \emph{predecessor}~\cite[Example~18]{Oost:08}), and
    $\bnat$ the maximal length of the development of the multistep;
  \item
    $\dtrs$-steps are labelled by any decreasing labelling, which exists
    by completeness of decreasing diagrams (Theorem~\ref{thm:dd:complete})
    for countable systems.
  \end{itemize}
  By $\succ$ we denote the well-founded order that orders
  triples for $\atrsp$-steps by the lexicographic product of 
  the greater-than relation $\snatgt$, $\aiarstra{\dtrs/\atrs}$,
  and of $\snatgt$ again, 
  the labels for $\dtrs$-step according to the decreasing labeling, 
  and the former above the latter.
  We show each local $\atrs$-peak can be completed into a decreasing
  diagram, distinguishing cases on steps and on whether their 
  redexes have overlap.
  \begin{itemize}
  \item
    A peak of shape $\relsco{\aiarsdevinv{\atrsp}}{\aiarsdev{\atrsp}}$
    such that the steps do not overlap can, by 
    Lemma~\ref{lem:horizontal}(\ref{ite:horizontal:non-overlap}),
    be completed by a valley of shape 
    $\relsco{\aiarsdev{\atrsp}}{\aiarsdevinv{\atrsp}}$.
    We conclude by a decrement in the first component
    for both multisteps in the valley.
  \item
    By Lemma~\ref{lem:horizontal}(\ref{ite:horizontal:overlap})
    an overlapping 
    peak of shape $\relsco{\aiarsdevinv{\atrsp}}{\aiarsdev{\atrsp}}$
    can be horizontally decomposed as
    $\relsco{\relsco{\aiarsdevinv{\atrsp}}{\relsco{\aiarsinv{\atrsp}}{\aiars{\atrsp}}}}{\aiarsdev{\atrsp}}$
    with $\relsco{\aiarsinv{\atrsp}}{\aiars{\atrsp}}$ an overlapping peak.
    By $\dtrs$ being critical-pair-closing for $\atrs$, that peak 
    can be closed by a $\dtrs$-conversion,
    so the original peak can be transformed into a conversion of shape 
    $\relsco{\relsco{\aiarsdevinv{\atrsp}}{\aiarsequ{\dtrs}}}{\aiarsdev{\atrsp}}$,
    which is seen to be decreasing: the first and second components of both $\atrsp$-multisteps 
    do not change, the first obviously so and the second by choosing to keep the same
    predecessors, but their third components decrease (by having developed one redex
    each), and $\dtrs$-steps are smaller than $\atrsp$-multisteps;
  \item
    A peak of shape $\relsco{\aiarsdevinv{\atrsp}}{\aiars{\dtrs}}$
    such that the steps do not overlap can, by 
    Lemma~\ref{lem:horizontal}(\ref{ite:horizontal:non-overlap}),
    be completed by a valley of shape 
    $\relsco{\aiarstre{\dtrs}}{\aiarsdevinv{\atrsp}}$,
    which is decreasing for the $\dtrs$-steps as these are
    by definition ordered below $\atrsp$-multisteps.
    To see decreasingness for the $\atrsp$-multistep,
    first observe that the first component does not change.
    Next, we distinguish cases on the type of the $\dtrs$-step.
    
    If it is a $\dtrsd$-step, i.e.\ it is duplicating, then by choosing its 
    source as second component it decreases.
    
    If it is a $\dtrsp$-step, i.e.\ it is linear, then by choosing to
    keep the same term as second component, all three components
    are the same, resulting in a decreasing diagram again;
  \item
    A peak of shape $\relsco{\aiarsdevinv{\atrsp}}{\aiars{\dtrs}}$
    such that the steps are overlapping can,
    by the special case of Lemma~\ref{lem:horizontal} where one of the multisteps is a single step,
    be horizontally decomposed as
    $\relsco{\relsco{\aiarsdevinv{\atrsp}}{\aiarsinv{\atrsp}}}{\aiars{\dtrs}}$
    with 
    $\relsco{\aiarsinv{\atrsp}}{\aiars{\dtrs}}$
    an overlapping peak.
    Since $\dtrs$ is by assumption a subsystem of $\atrs$, 
    that peak
    is an $\atrs$-peak
    and we may proceed as in the second item. \qed
  \end{itemize}
\end{proof}

\begin{proof}[of Lemma~\ref{lem:hot:structural}]
  \begin{enumerate}
  \item
  Specialising Definition~\ref{def:decreasing} to the hot-labelling,
  hot-decreasingness of $\iarsdevinva{\alab}{\btrm}{\iarsdeva{\blab}{\atrm}{\ctrm}}$
  yields a conversion of shape 
   $\iarsequa{\ordhotdwna{\alab}
            }{\btrm
  }{\iarsdeva{\blab
            }{\btrm'
  }{\iarsequa{\ordhotdwna{\alab\blab}
            }{\btrm''
}{\iarsdevinva{\alab
            }{\ctrm''
  }{\iarsequa{\ordhotdwna{\blab}}
            }{\ctrm'
            }{\ctrm
            }}}}$,
  and similarly for the peak vector, 
  where we have used that multisteps may be empty so that 
  $\releq{\aarsdev}{\relref{\aarsdev}}$,
  and that due to the properties of $\aordmul$, if $\stple{\bStp}{\aStp}$ then
  $\ordrefhota{\Labhota{\aStp}}{\Labhota{\bStp},\Labhota{\stpres{\aStp}{\bStp}}}$
  so that any multistep $\aiarsdev{\alab}$ can be 
  developed into a reduction $\aiarstre{\ordhotrefdwna{\alab}}$ of ordinary steps.\footnote{%
We may even assume~\cite{Felg:15} the multisteps are \emph{homogeneous} (all rules the same label).}

   That all variables (which are $1$st-order) occurring in the diagram
   may be assumed to be contained in the variables occurring in $\atrm$, say $\vec{\ctrmvar}$, follows by 
   simply substituting the same constant\footnote{%
If there is no constant in the signature, as fresh constant may be adjoined without
affecting confluence, as confluence is a modular property of TRSs.}
\footnote{%
Of course, if we already know that \emph{all} peaks can be completed
into a decreasing diagram, then its is obvious, because
then the system is confluent.
}
  for all variables not among $\vec{\ctrmvar}$ in the diagram.
  Since steps, conversions, and multisteps are closed under substitution, this preserves the shape
  of the diagram, and it even does not change the peak at all:
  as $\btrm$,$\ctrm$ are obtained from $\atrm$ by rewriting, and rewrite rules are 
  assumed not to introduce variables, their variables are among those of $\atrm$.
  To see that the diagram is still hot-decreasing, note that 
  if a term $\etrm$ on the closing conversion is the source of step 
  having a label required to be 
  $\aordhot$-ordered below a \emph{source} label (i.e.\ term) of one 
  of the steps $\aStp$, $\bStp$, i.e.\ below $\atrm$,
  then $\arstrea{\atrm}{\etrm}$ from which we conclude that the variables contained 
  in $\etrm$ are a subset of $\vec{\ctrmvar}$,
  and if it was required to be below a \emph{rule} label (i.e.\ set) of
  one of these steps, then we conclude since those labels and their order $\aordmul$ 
  are invariant under substitution of constants.
  \item    
  We show the valley $\relap{\relsco{\aiarsdev{\stpres{\bStp}{\aStp}}
        }{\aiarsdevinv{\stpres{\aStp}{\bStp}}
         }
        }{\btrm
        }{\ctrm
        }$ completes the peak
  completes the peak $\iarsdevinva{\aStp}{\btrm}{\iarsdeva{\bStp}{\atrm}{\ctrm}}$
  into a hot-decreasing diagram.
  by considering all possible distributions of $\dtrs$- and 
  $\setdif{\atrs}{\dtrs}$-rules in $\aStp$, $\bStp$.
  By Lemma~\ref{lem:horizontal}(\ref{ite:horizontal:non-overlap}),
  the rule symbols occurring in 
  $\stpres{\bStp}{\aStp}$ are contained in $\bStp$,
  and those in $\stpres{\aStp}{\bStp}$ are contained in $\aStp$.
  We have on the one hand that if $\aStp$
  contains some rule in $\setdif{\atrs}{\dtrs}$ then 
  $\ordhotrefa{\Labhota{\aStp}}{\Labhota{\stpres{\aStp}{\bStp}}}$
  since either $\stpres{\aStp}{\bStp}$ contains such a rule 
  as well so their sets of maxima are $\aordmulref$-related,
  or else we conclude by $\aordhot$ ordering sets above terms.
  On the other hand, if $\aStp$ only contains $\dtrs$-rules then
  so does $\stpres{\aStp}{\bStp}$ and either
  $\bStp$ contains some rule in $\setdif{\atrs}{\dtrs}$ and then
  $\ordhota{\Labhota{\aStp}}{\Labhota{\stpres{\aStp}{\bStp}}}$,
  or it does not and then 
  $\ordhotrefa{\Labhota{\aStp}}{\Labhota{\stpres{\aStp}{\bStp}}}$
  as their sources are $\aiarstre{\dtrs}$-related, as desired.
  \item
  We distinguish cases on the types of the rules in the composite peak 
  \[\iarsdevinva{\alab'
              }{\substra{\vec{\atrmvar}}{\vec{\btrm}}{\btrm}
    }{\iarsdeva{\blab'
              }{\substra{\vec{\atrmvar}}{\vec{\atrm}}{\atrm}
              }{\substra{\vec{\atrmvar}}{\vec{\ctrm}}{\ctrm}
              }}\] having labels as indicated.
  \begin{itemize}
  \item
    If either of the multisteps is empty, we conclude trivially;
  \item
    If both multisteps only contain $\dtrs$-rules, then first note that the $\dtrs$-conversions
    may be further restricted to be of shape $\iarsequa{\ordhotdwna{\atrm}}{\btrm}{\ctrm}$,\footnote{%
The conversion is \emph{below} $\atrm$ (with respect to $\aiarstra{\dtrs}$)
in the sense of~\cite{Wink:Buch:83,Oost:08}.
}   
    by using that a step $\arsa{\atrm}{\ldots}$ cannot occur in it.
    This is seen by considering what would be to the left of such a step in the conversion:
    it cannot be the first step since $\trmneq{\atrm}{\btrm}$ by termination of $\dtrs$;
    it cannot be preceded by a step $\arsinva{\ldots}{\atrm}$ as that would have label
    $\atrm$, not a smaller one as required by decreasingness;
    and not by a step $\arsa{\ldots}{\atrm}$ as its source cannot be smaller than $\atrm$
    as then $\aiars{\dtrs}$ would be cyclic.    
    Based on this, we construct the closing conversion\footnote{%
The notation, substituting conversions at parallel positions, 
leaves unspecified the (sequential) order of the steps of the 
conversions substituted. Any choice will do.
}
    \[ \iarsequa{\ordhotrefdwna{\substra{\vec{\atrmvar}}{\vec{\atrm}}{\btrm}}
             }{\substra{\vec{\atrmvar}}{\vec{\btrm}}{\btrm}
   }{\iarsequa{\ordhotdwna{\substra{\vec{\atrmvar}}{\vec{\ctrm}}{\atrm}}
             }{\substra{\vec{\atrmvar}}{\vec{\ctrm}}{\btrm}
             }{\substra{\vec{\atrmvar}}{\vec{\ctrm}}{\ctrm}
             }} \]
     It is decreasing as 
     $\ordhota{\substra{\vec{\atrmvar}}{\vec{\atrm}}{\atrm}
             }{\substra{\vec{\atrmvar}}{\vec{\atrm}}{\btrm}
             }$ by 
     $\ordhota{\atrm}{\btrm}$ and closure of non-empty $\atrs$-reductions under substitution,
     and
     $\ordhotrefa{\substra{\vec{\atrmvar}}{\vec{\atrm}}{\atrm}
                }{\substra{\vec{\atrmvar}}{\vec{\ctrm}}{\atrm}
                }$ by $\ordhota{\vec{\atrm}}{\vec{\ctrm}}$ and
                closure of $\atrs$-reductions under contexts.
  \item
    If both contain some $(\setdif{\atrs}{\dtrs})$-rule, then we conclude by the conversion
    \[\scriptstyle 
    \iarsequa{\ordhotdwna{\alab\vec{\alab}}
            }{\substra{\vec{\atrmvar}}{\vec{\btrm}}{\btrm}
  }{\iarsdeva{\ordhotrefdwna{\blab'}
            }{\substra{\vec{\atrmvar}}{\vec{\btrm'}}{\btrm'}
  }{\iarsequa{\ordhotdwna{\alab\vec{\alab}\blab\vec{\blab}}
            }{\substra{\vec{\atrmvar}}{\vec{\btrm''}}{\btrm''}
}{\iarsdevinva{\ordhotrefdwna{\alab'}
            }{\substra{\vec{\atrmvar}}{\vec{\ctrm''}}{\ctrm''}
  }{\iarsequa{\ordhotdwna{\blab\vec{\blab}}}
            }{\substra{\vec{\atrmvar}}{\vec{\ctrm'}}{\ctrm'}
            }{\substra{\vec{\atrmvar}}{\vec{\ctrm}}{\ctrm}
            }}}}\] 
    obtained by piecewise composing the constituents conversions,
    which is decreasing because 
    $\setge{\setlub{\alab}{\vec{\alab}}}{\ordmulrefa{\alab'}{\alab,\vec{\alab}}}$ and
    $\setge{\setlub{\blab}{\vec{\blab}}}{\ordmulrefa{\blab'}{\blab,\vec{\blab}}}$.
  \item
    If one of them, say $\aStp$, only contains $\dtrs$-steps but the other doesn't.
    then the constituent conversions have shape
    $\iarsequa{\ordhotdwna{\atrm}
            }{\btrm
  }{\iarsdeva{\blab
            }{\btrm'
  }{\iarsequa{\ordhotdwna{\blab}
            }{\ctrm'
            }{\ctrm
            }}}$ yielding
   \[\iarsequa{\ordhotdwna{\substra{\vec{\atrmvar}}{\vec{\atrm}}{\atrm}}
            }{\substra{\vec{\atrmvar}}{\vec{\btrm}}{\btrm}
  }{\iarsdeva{\ordhotrefdwna{\blab'}
            }{\substra{\vec{\atrmvar}}{\vec{\btrm'}}{\btrm'}
  }{\iarsequa{\ordhotdwna{\blab\vec{\blab}}
            }{\substra{\vec{\atrmvar}}{\vec{\ctrm'}}{\ctrm'}
            }{\substra{\vec{\atrmvar}}{\vec{\ctrm}}{\ctrm}
            }}}\]
    which is seen to be decreasing by reasoning as in the previous items. 
  \end{itemize}
  \end{enumerate}
\end{proof}

\begin{proof}[of measure decrease in the induction step of Theorem~\ref{thm:hot}]
    We have $\seteq{\cctx'}{\src{\bStp'}}$ hence 
   $\seteq{\cctx}{\clslub{\cctx'}{\src{\bstp}}}$
   as sets of patterns despite $\trmneq{\atrm}{\ctrm'}$;
   the positions are in both since $\bstp$ is inner so the 
   patterns $\vec{\blhs}'$ are not below $\bmtrmvar$ in $\bmtrm$.\footnote{%
In the let-expression representation this follows from Proposition~\ref{prop:multipattern:size} by
vertically decomposing both having as substitute the redex respectively
the contractum of $\bstp$.}
    Using distributivity all the time, 
    the strict inequality holds in the line of reasoning~(\ref{eqn:decrease}) by 
    $\setneq{\setglb{\src{\astp}}{\src{\bstp}}}{\seteq{\setemp}{\setglb{\src{\bstp}}{\dctx}}}$,
    the first equality by the reasoning above and $\pair{\astp}{\bstp}$ being inner so that 
    $\clseq{\clsglb{\bictx{-}}{\src{\bstp}}}{\clsbot}$,
    the second since, by reasoning as for $\cctx$, we have $\clseq{\bctx'}{\clslub{\src{\aStp'}}{\dctx'}}$
    since the patterns $\vec{\alhs}'$ are not below $\cmtrmvar$ in $\cmtrm$
    by $\astp$ being inner, hence 
    $\clseq{\bictx{-}
   }{\clseq{\clslub{\src{\aStp'}}{\dctx}
          }{\bictx{+}'
          }}$, and the final inequality by
    $\clsge{\clsglb{\dctx}{\src{\bStp'}}}{\clsglb{\dctx'}{\src{\bStp'}}}$ which
    holds because by $\hat{\bStp}$ being a multistep from $\hat{\ctrm}$,
    $\dctx'$contains positions that are either below the root of $\bstp$ but then not in $\bStp'$, 
    or in the pattern of $\astp$ and then in $\dctx$.

  We now give the idea how, as an alternative to the set-theoretic reasoning.
  one can also directly work on let-expressions 
  to show that the induction measure decreases,  
  i.e.\ one can proceed by giving  appropriate \emph{constructions} on multipatterns,
  and then showing that the measure decreases, constructing witnesses by \emph{computation}.
  
    For instance, one may define $\bictx{-}$ as in the main text,
    but now by a let-expression, as
    $\clslet{\vec{\amtrmvar'}\cmtrmvar'}{\vec{\alhs}'\ctxa{\ctrmvar'}
         }{\substra{\cmtrmvar}{\cmtrmvar'(\vec{\ailhs{\astp}})}{\cmtrm}}$
    for $\cmtrmvar'$, $\ctrmvar'$ fresh. Here $\actx$ \emph{is} $\cctx$ as defined above,
    but now \emph{constructed} from the image $\ctxa{\mtrmvarb{\vec{\ailhs{\astp}}}}$ 
    of $\cmtrmvar$ under the $2$nd-order substitution $\bsub$ witnessing 
    $\clsle{\src{\bstp}}{\acls}$, where $\acls$ in turn was constructed as the join of $\bstp$ and $\aStp$.
    For another example, $\dctx$, the part of the pattern of $\astp$ that 
    does not belong to the pattern of $\astp$, can be constructed by
    $\clslet{\cmtrmvar'}{\ctxa{\ctrmvar'}
         }{\substra{\vec{\amtrmvar'}\cmtrmvar}{\vec{\alhs}'\cmtrmvar'(\vec{\ailhs{\astp}})}{\cmtrm}}$.

    The same reasoning applies, but now by \emph{computation} on multipatterns.
    For instance, using distributivity to decompose $\aStp$, $\bStp$ in their constituent steps,
    the inequality on the amount of overlap follows from
    $\clspsz{\clsglb{\subap{\sstpsrc}{\astp}
          }{\clsglb{\subap{\sstpsrc}{\bstp}}}} > 
     \clspsz{\clsbot}   =   
     \clspsz{\clsglb{\subap{\sstpsrc}{\astp}
                   }{\clsbot
                   }}   = 
    \clspsz{\clsglb{\subap{\sstpsrc}{\astp}
           }{\clsglb{\subap{\sstpsrc}{\bstp}
                   }{\bictx{-}
                   }}
           }$
    with $\clseq{\clsglb{\subap{\sstpsrc}{\bstp}}{\bictx{-}}}{\clsbot}$ 
    by their complementary definition via $\actx$.
    Similarly, one may proceed from the right.
    
    A difference in the reasoning shows up `in the middle' of the line
    of reasoning~(\ref{eqn:decrease}):
    since let-expression can only be compared with respect to the
    refinement order $\sclsle$ if they denote the \emph{same} term,
    peaks for different terms, as is the case here, can a priori not be compared.
    The way around this is to use Proposition~\ref{prop:multipattern:size} 
    to split-off any differing (non-patterns) parts first.
    For instance, the amounts of overlap
    $\pkovl{(\clslet{\amtrmvar}{f(f(x))}{\mtrmvara{f(a)}})
          }{(\clslet{\bmtrmvar}{f(f(x))}{f(\mtrmvarb{a})})
          }$ and 
     $\pkovl{(\clslet{\amtrmvar}{f(f(x))}{\mtrmvara{f(b)}})
           }{(\clslet{\bmtrmvar}{f(f(x))}{f(\mtrmvarb{b})})
           }$ are clearly the  same; the subterms $a$ and $b$ are innocuous here.
     That can be implemented for let-expressions 
     by vertically decomposing the let-expression involved. For instance,
     decomposing the first as
     $\substra{\atrmvar
             }{(\clslet{}{}{a})
             }{(\clslet{\amtrmvar}{f(f(x))}{\mtrmvara{f(\atrmvar)})}
             }$. Since the total amount of overlap 
     is the sum of that of the components, 
     the substitutes have no overlap, and the
     prefixes now have the same denotations, we may proceed,
     and have recovered the possibilities of the set-theoretic 
     representation on let-expressions. \qed
   \end{proof}     
 
\fi

\end{document}